\documentclass[12pt,english]{article}

\usepackage{geometry}
\usepackage[english]{babel}
\usepackage{setspace}
\usepackage{amsmath}
\usepackage{amssymb}
\usepackage{amsthm}
\usepackage{booktabs}
\usepackage{graphicx}
\usepackage{subcaption}
\usepackage[round]{natbib}
\usepackage{bibunits}
\usepackage{hyperref}
\IfFileExists{to_do_list.tex}{
    \newcommand{\paperpath}{}
    \defaultbibliography{References/ref}
}{
    \newcommand{\paperpath}{Paper/}
    \defaultbibliography{Paper/References/ref}
}
\graphicspath{{\paperpath}}
\hypersetup{
    colorlinks=true,
    citecolor=blue,
    linkcolor=black,
    urlcolor=blue
}

\defaultbibliographystyle{plainnat}

\newcommand{\figurenote}[1]{%
    \par\vspace{0.4em}%
    \begin{minipage}{0.92\textwidth}%
        \footnotesize\emph{Notes:} #1%
    \end{minipage}%
}

\newtheorem{assumption}{Assumption}[section]
\newtheorem{proposition}{Proposition}

\title{Identifying the MPC-Liquidity Gradient in High-Quality Data\thanks{We thank Jeppe Druedahl, Johan Gustafsson, Hannes Malmberg, Christian Matthews, Chris Moser, Azeem Shaikh, Fredrik Sävje and Gianluca Violante for helpful discussions. We are also thankful for comments from participants at multiple seminars and conferences. We are grateful for financial support from Riksbankens Jubileumsfond (program M23-0019) and Handelsbankens Forskningsstiftelse (project P25-0018). The opinions expressed in this article are the sole responsibility of the authors and should not be interpreted as reflecting the views of Sveriges Riksbank.}}
\author{Mikael Carlsson\textsuperscript{\dag}, Marco D'Amico\textsuperscript{\ddag}, Erik Öberg\textsuperscript{\S},\\ Oskar N. Skans\textsuperscript{\P}, Karl Walentin\textsuperscript{$\parallel$}}
\date{\today}

\begin{document}
\begin{bibunit}

\clearpage\maketitle
\begingroup
\renewcommand{\thefootnote}{\fnsymbol{footnote}}
\footnotetext[2]{\scriptsize \textsuperscript{\dag} Uppsala University, UCLS and Sveriges Riksbank, mikael.carlsson@nek.uu.se; \textsuperscript{\ddag} Uppsala University, marco.damico@nek.uu.se; \textsuperscript{\S} Uppsala University, University of Oslo, UCLS and CESifo, erik.oberg@nek.uu.se; \textsuperscript{\P} Uppsala University and UCLS, oskar.nordstrom\_skans@nek.uu.se; \textsuperscript{$\parallel$} Uppsala University and Sveriges Riksbank, karl.walentin@nek.uu.se.}
\endgroup
\thispagestyle{empty}

\begin{abstract}
\begin{spacing}{0.95}
\noindent
We estimate the gradient of the Marginal Propensity to Consume (MPC) with respect to liquidity using a new estimator designed for administrative data with negligible measurement error in income. 
We derive a state-dependent consumption pass-through equation from the canonical buffer-stock model, and show that the pass-through coefficient of this equation can be used to construct tight bounds on the MPC conditional on the relevant state. 
We recover latent permanent and transitory income shocks with a Kalman smoother and use them as regressors in an empirical representation of the consumption equation.
The Kalman-shock estimator identifies the theoretical pass-through coefficient under the assumption of negligible measurement error in income, and attains the minimum variance within the class of estimators that are linear in household income histories, including the canonical GMM estimator by \citet{Blundell_et_al_2008_AER} and recent refinements thereof.
Applying the method to Swedish administrative tax registers, we show that consumption responses to transitory shocks have a sharp negative and convex gradient with respect to cash-on-hand; the associated annual MPC falls from 0.7 in the lowest cash-on-hand decile to 0.3 in the top decile.
The permanent-shock pass-through is close to one across the cash-on-hand distribution.
These patterns are not visible when using traditional, less efficient, estimators.
\end{spacing}
\end{abstract}


\section{Introduction}
The marginal propensity to consume (MPC) is a key link between household income fluctuations and aggregate demand. It matters for the size of fiscal multipliers, the transmission of monetary policy, and the amplification of macroeconomic shocks.
It is also a key statistic for distinguishing consumption-savings theories and for disciplining heterogeneous-agent models used for counterfactual policy analysis.
As a consequence, a rapidly evolving literature has focused on documenting the level and distribution of the MPC. 

A central branch of this literature employs the ``semi-structural'' approach, pioneered by \citet{Blundell_et_al_2008_AER} (henceforth BPP), and later refined by \citet{Commault2022} (henceforth BPP-C). The key insight in these papers is that, under a linear consumption rule, the pass-through of transitory income shocks to consumption can be estimated using linear combinations of the household's income history (``covariance restrictions'').
These estimators were developed primarily for survey data where income is measured with error and where sample sizes limit the scope for estimating heterogeneous treatment effects. In that environment, robustness to measurement error is central.

Recent papers apply these semi-structural estimators to large administrative datasets to estimate MPCs and MPC gradients, including \citet{Ganong2025}, \citet{FlorinOBilbiie2025}, and \citet{Crawley2023}. In these settings, income histories are precisely measured, and sample sizes are large enough to study heterogeneity directly. Robustness to severe income measurement error remains useful, but efficiency and state dependence become more valuable. 

This paper adapts the semi-structural approach to a modern data-rich environment. First, we derive a state-dependent pass-through equation from a canonical buffer-stock consumption problem, and show how the pass-through coefficient bounds the MPC conditional on the relevant state. The state variables include cash-on-hand normalized by permanent income, which we take as our measure of liquidity. Second, we show that Kalman-smoothed shocks identify the pass-through coefficient and give the minimum-variance signal in the class of identifying signals that are linear in the household's income history, which includes BPP and BPP-C. Third, we apply the estimator to Swedish administrative data on income, wealth, and budget-constraint-imputed consumption. The main empirical result is a sharp and convex negative MPC gradient with respect to normalized cash-on-hand: MPCs fall especially quickly at low levels of cash-on-hand and then flatten out higher in the distribution. A traditional BPP-style estimator is too imprecise to detect this gradient in our data.

Our theoretical contribution is to connect state-dependent pass-through coefficients to MPCs, building on the consumption-growth approximation of \citet{Blundell_Low_Preston_2013_QE} and \citet{Arellano2017}. We start from a permanent-transitory MA(k) income process and a buffer-stock consumption-savings problem in the tradition of \citet{Deaton1991} and \citet{Carroll1997}. A linear approximation of the consumption policy gives the pass-through equation. This equation generalizes the BPP specification by allowing pass-through coefficients to vary with predetermined household states, especially cash-on-hand normalized by permanent income. The coefficient on a transitory income innovation cannot mechanically be translated into the MPC with respect to current income, because such innovations may also contain news about future income whenever $k>0$. However, when the persistence of transitory income shocks is modest, as in our data, the pass-through coefficient delivers tight bounds on the MPC with respect to current cash-on-hand. The framework also accommodates liquid/illiquid-asset models with non-convex adjustment costs, provided that the consumption equation is estimated conditional on no adjustment in the illiquid account, which we observe in our data.

Our estimator of the state-dependent pass-through coefficient is based on retrieving the latent permanent and transitory income shocks using a Kalman smoother, inspired by recent work applying such smoothers to recover shocks from income data \citep{Braxton2025}. Relative to existing estimators, the advantage of this approach is that we use the full income history, both past and future, to recover household-level income shocks. Specifically, we first estimate the parameters of a common permanent-transitory income process from pooled income-growth moments. We then use the Kalman smoother to form linear predictions of the latent permanent and transitory shocks. Finally, we regress consumption growth on the recovered shocks, allowing the coefficients to vary with predetermined state variables.

This procedure is especially well suited for high-quality administrative income data. Under negligible income measurement error, the Kalman-smoothed shock is the best linear prediction of the latent shock from the household's income history, even when the underlying shock distributions are leptokurtotic (as emphasized in \citet{Guvenen2021}). This projection identity implies that the population coefficient on the recovered transitory shock identifies the pass-through coefficient and that the Kalman-smoothed signal has minimum variance within the class of identifying signals that are linear in household income histories, a class that includes the BPP and the refined BPP-C estimator. Notably, however, with non-negligible income measurement error, the Kalman smoother will not be able to distinguish noise from transitory income shocks, and our approach will therefore be plagued by attenuation bias when applied to survey data. 

We apply the estimator to data from the Swedish administrative tax registers for 2000--2007, covering the universe of tax-paying individuals.\footnote{The same data have been used to study consumption behavior in \citet{Floden_Kilstrom_Sigurdsson_Vestman2021_EJ} and \citet{DiMaggio_Kermani_Majlesi2020_JF}.} The data contain income, transfers, demographics, and wealth, and allow us to impute consumption expenditures from the household budget constraint, following \citet{BrowningLethPetersen2003} and \citet{Kolsrud_et_al_2020_JPUB}. Our baseline sample focuses on singles, which avoids complications in estimating household income processes. The results are, however, robust to including multi-member households.

The estimates show a steep and convex cash-on-hand gradient. The average yearly MPC is about 0.5, ranging roughly 0.7 in the bottom decile of the cash-on-hand distribution to roughly 0.3 in the top decile, with most of the decline concentrated at low levels of normalized cash-on-hand. Because these spending responses include some durable purchases, such as cars, the magnitudes are not far from what one would expect from a calibrated two-asset buffer-stock savings model, widely used in the macroeconomic literature. When the MPC gradient by cash-on-hand is decomposed into income and liquid financial wealth components, liquid financial wealth emerges as the primary driver.  

The efficiency property of our estimator is essential for detecting the empirical gradient. We show this by estimating the same cash-on-hand gradient in the MPC using the BPP-C estimator. This estimator, as the original BPP estimator, sacrifices efficiency to ensure robustness to measurement error in income and misspecification of the consumption policy. It achieves this by only exploiting the covariance between current consumption growth and distant future income growth, instead of using the full income history as we do. Because of limited persistence in the transitory shocks, the BPP-C estimator is much less precise than our estimator and its bootstrapped standard errors are very large. As a consequence, it fails to detect the gradient in our data. Moreover, while the BPP-C estimates are sensitive to the assumed lag length of the income process, our estimates remain stable and precise. 

We also present results for the consumption response to permanent income shocks. The same logic underpinning the identification and efficiency results for the transitory pass-through coefficient also applies to the permanent-shock pass-through coefficient. Also for this parameter, the best linear projection of the permanent shock on the income history is the Kalman-smoothed permanent shock. Here, we find that the consumption response is flat and centered around one over the cash-on-hand distribution. This pattern is also consistent with the buffer-stock model interpretation.

Broadly, our paper extends the large and long-standing research program on testing models of consumption-savings behavior \citep{ModiglianiBrumberg1954,Friedman1957,Hall1978,Deaton1991}. The evidence, showing a steep liquidity gradient of consumption responses with respect to transitory income shocks, but a flat gradient with respect to permanent shocks, is consistent with the canonical incomplete-markets model used widely in the household finance and macroeconomic literatures. 

More narrowly, we make three distinct contributions. First, in terms of methodology, the paper extends the literature on semi-structural estimators for consumption functions associated with \citet{Blundell_et_al_2008_AER}, \citet{Arellano2017}, \citet{Commault2022}, and \citet{Crawley2023}. We provide a refined methodology for recovering state-dependent pass-through coefficients and MPCs in data with little measurement error in income. Second, it complements natural-experiment evidence on MPCs using tax rebates, lottery wins, and unemployment spells \citep{Johnson_et_al_2006_AER,Parker_2013_AER,Misra2014,Fagareng_et_al_2021_AEJM,Patterson2023,Fagereng2024_et_al_JME,Golosov2024,Boehm2025}. Those designs are powerful, but they often identify responses to particular shocks or among selected populations, and it is not clear to what extent they carry over to other settings. Many of these papers do not detect a clear gradient with respect to liquid assets. In contrast, our estimates use all yearly variation in earnings, taxes, and transfers and we obtain a significant and steep cash-on-hand gradient. Third, the results provide empirical discipline for the liquidity/MPC channel in heterogeneous-agent macroeconomic models (see, e.g., \citet{Kaplan2022} and \citet{Auclert2025}).

The rest of the paper proceeds as follows. Section \ref{sec:identifying_assumptions} derives the pass-through equation and MPC bounds. Section \ref{sec:estimator} presents the estimator and its properties in terms of identification, consistency and efficiency. Section \ref{sec:data} describes the Swedish administrative data. Section \ref{sec:empirical_results} presents the results. The final section concludes.

\section{Identifying Assumptions}
\label{sec:identifying_assumptions}
In this section, we describe our empirical target and the basic assumptions that allow us to estimate the MPC. We first derive a state-dependent pass-through equation as a linear approximation of the canonical buffer-stock consumption model \citep{Deaton1991,Carroll1997}, and then state the additional restrictions that allow us to map a properly estimated consumption effect of transitory income innovations into MPC bounds. Our approach builds on the linear approximation in \citet{Blundell_Low_Preston_2013_QE}. The purpose of grounding the pass-through equation in the buffer-stock consumption model is not to impose a parametric consumption rule. It is to discipline which state variables are needed to interpret the local response to a new income innovation.

\begin{assumption}[Income process]
\label{ass:income_process}
For individual $i$ in period $t$, income ($Y_{i,t}$) is the product of a deterministic component ($Y^d_{i,t}$), a permanent stochastic component ($P_{i,t}$), and a transitory stochastic component ($\exp(\nu_{i,t})$):
\begin{align}
    Y_{i,t} &= Y^d_{i,t} P_{i,t}\exp(\nu_{i,t}), \label{eq:income_process_level} \\
    P_{i,t} &= P_{i,t-1}\exp(\eta_{i,t}), \label{eq:permanent_process} \\
    \nu_{i,t} &= \sum_{j=0}^{k}\theta_j\varepsilon_{i,t-j}, \qquad \theta_0=1. \label{eq:transitory_process}
\end{align}
The permanent innovation ($\eta_{i,t}$) and transitory innovation ($\varepsilon_{i,t}$) are mean zero with finite variances, the transitory component follows an MA($k$) process with coefficients $\theta_j$, and the income-process parameters are common across households.
\end{assumption}

Assumption \ref{ass:income_process} follows the permanent-transitory representation used in the semi-structural MPC literature. Permanent shocks shift the level of future income, while transitory shocks have an MA$(k)$ effect on income. The common-parameter restriction is what makes it possible to estimate the income process from pooled income-growth moments and then recover household-level shocks.

\begin{assumption}[Budget set]
\label{ass:budget_set}
Households allocate current cash-on-hand ($M_{i,t}$) between consumption ($C_{i,t}$) and next-period assets ($A_{i,t+1}$):\footnote{Here we handle the interest rate $r$ as constant. Note that all aggregate time-varying features, including the interest rate that in theory should have a (minor) impact on the MPC, are neutralized in the empirical work through year fixed effects and the fact that our deciles are defined through year-specific cash-on-hand boundaries. Thus, our estimates of each decile-specific MPC capture the average MPC for each decile across all years in our data.}
\begin{align}
    C_{i,t}+A_{i,t+1} &= M_{i,t}, \label{eq:budget_constraint} \\
    M_{i,t} &= (1+r)A_{i,t}+Y_{i,t}. \label{eq:cash_on_hand}
\end{align}
\end{assumption}

\begin{assumption}[Consumption behavior]
\label{ass:consumption_behavior}
Consumption is given by a normalized consumption policy ($c$):
\begin{equation}
    C_{i,t}=P_{i,t}c(m_{i,t},\mathbf{s}_{i,t})
    \label{eq:consumption_policy}
\end{equation}
where
\begin{equation}
    c(m_{i,t},\mathbf{s}_{i,t})\equiv \frac{C_{i,t}}{P_{i,t}},
    \qquad 
    m_{i,t} \equiv  \frac{M_{i,t}}{P_{i,t}},
    \qquad
    \mathbf{s}_{i,t} \equiv  (\varepsilon_{i,t},\ldots,\varepsilon_{i,t-k+1}).
    \label{eq:state_variables}
\end{equation}
\end{assumption}

This consumption policy is the standard implication of a buffer-stock savings model with CRRA preferences, a linear budget constraint, and borrowing constraints, as in \cite{Deaton1991} and \cite{Carroll1997}. Consumption scales linearly with permanent income, holding the state variables constant. The state variables are cash-on-hand normalized by permanent income ($m_{i,t}$), and the transitory-income state vector ($\mathbf{s}_{i,t}$) contains the recent transitory innovations that help forecast future income.

In a two-asset model with both liquid wealth and illiquid wealth \citep[as in][]{Kaplan_Violante2014}, commonly used in the modern macroeconomic literature, the state space is larger. In these settings, households also decide whether to adjust illiquid positions, and consumption responses are typically very different conditional on adjustment relative to non-adjustment. Our empirical implementation deals with this by focusing on periods without housing transactions, which is the main component of illiquid wealth in Sweden. 
In the extended two-asset version of the Deaton-Carroll buffer-stock consumption model, the relevant state variables conditional on not adjusting the illqiuid position are still normalized cash-on-hand and the transitory income state vector, as assumed here.

Combining Assumptions \ref{ass:income_process}--\ref{ass:consumption_behavior} gives the pass-through equation. Taking logs and first-differencing the consumption policy gives
\begin{equation}
    \Delta \log C_{i,t}
    =
    \eta_{i,t}
    +\log c(m_{i,t},\mathbf{s}_{i,t})
    -\log c(m_{i,t-1},\mathbf{s}_{i,t-1}).
    \label{eq:consumption_growth_policy}
\end{equation}
The budget constraint implies the law of motion
\begin{equation}
    m_{i,t}
    =
    \frac{1+r}{\exp(\eta_{i,t})}
    \left(m_{i,t-1}-c(m_{i,t-1},\mathbf{s}_{i,t-1})\right)
    +Y^d_{i,t}\exp(\nu_{i,t}).
    \label{eq:normalized_cash_on_hand_lom}
\end{equation}
Conditioning on the predetermined observed state $(m_{i,t-1},\mathbf{s}_{i,t-1})$ and linearizing with respect to the new innovations $(\eta_{i,t},\varepsilon_{i,t})$ gives the following result.

\begin{proposition}[Pass-through equation]
\label{prop:passthrough}
Under Assumptions \ref{ass:income_process}--\ref{ass:consumption_behavior}, the first-order approximation to consumption growth around $\eta_{i,t}=\varepsilon_{i,t}=0$ is
\begin{equation}
    \Delta \log C_{i,t}
    \approx
    \lambda(m_{i,t-1},\mathbf{s}_{i,t-1})\eta_{i,t}
    +
    \gamma(m_{i,t-1},\mathbf{s}_{i,t-1})\varepsilon_{i,t}.
    \label{eq:passthrough_equation}
\end{equation}
The coefficient $\lambda$ is the pass-through of permanent income shocks. The coefficient $\gamma$ is the pass-through of transitory income shocks and can be written as
\begin{equation}
    \gamma(m_{i,t-1},\mathbf{s}_{i,t-1})
    =
    Y^d_{i,t}\exp(\nu^0_{i,t})
    \frac{\partial \log c(m_{i,t},\mathbf{s}_{i,t})}{\partial m_{i,t}}
    +
    \frac{\partial \log c(m_{i,t},\mathbf{s}_{i,t})}{\partial s^{(1)}_{i,t}},
    \label{eq:gamma_cash_news}
\end{equation}
where $\nu^0_{i,t}\equiv\sum_{j=1}^{k}\theta_j\varepsilon_{i,t-j}$ is the predetermined transitory component and $s^{(1)}_{i,t}\equiv\varepsilon_{i,t}$ is the current transitory innovation.\footnote{We have suppressed the dependence of $\gamma$ on $Y^d_{i,t}$ here as we remove the deterministic part of ("residualize") log income and log consumption in our empirical application.}
\end{proposition}

\noindent\textit{Proof:} See Appendix Section~\ref{app:identification_proofs}.

Equation \eqref{eq:passthrough_equation} is a state-dependent generalization of the consumption-growth equation in \cite{Blundell_et_al_2008_AER}. In their specification, log consumption growth similarly depends linearly on permanent and transitory income shocks, assumed to be common for the entire population. Here the same two shocks enter consumption growth, but the coefficients are functions of the household's predetermined state. Note that the pass-through coefficients $\gamma$ and $\lambda$ are elasticities, and not dollar-to-dollar MPCs. The conversion factor between the two is the consumption-to-income ratio $\frac{C_{it}}{Y_{it}}$.

The expression for $\gamma$ in Equation \eqref{eq:gamma_cash_news} illustrates that this pass-through coefficient is a composite of two effects. The first term is the cash-on-hand channel: a transitory innovation raises current resources. The second term is a news channel: the same innovation may also change expected future income through the MA component. Estimating $\gamma$ identifies the total response to the transitory innovation. To recover the latent information in this coefficient about the response as an MPC with respect to current cash-on-hand, we need to bound the importance of the news channel. We do this under the assumption that increased income today does not decrease income in the future,     $\left(\sum_{j=1}^{k}\theta_j\right)\ge 0$.

\begin{assumption}[Bounded response to news]
\label{ass:bounded_news}
For $\left(\sum_{j=1}^{k}\theta_j\right) \ge 0$, for all states $(m_{i,t},\mathbf{s}_{i,t})$, 
\begin{equation}
    0
    \le
    \frac{\partial \log c(m_{i,t},\mathbf{s}_{i,t})}{\partial s^{(1)}_{i,t}}
    \le
    \left(\sum_{j=1}^{k}\theta_j\right)
    Y^d_{i,t}\exp(\nu^0_{i,t})
    \frac{\partial \log c(m_{i,t},\mathbf{s}_{i,t})}{\partial m_{i,t}}.
    \label{eq:bounded_news}
\end{equation}
\end{assumption}

Assumption \ref{ass:bounded_news} states that good news about future income does not raise current consumption more than receiving the corresponding income today. This is a weak implication of standard consumption-savings logic: cash-on-hand can be consumed immediately, while news about future income must be smoothed through the intertemporal budget constraint. The assumption becomes less restrictive when the transitory component is only weakly persistent, because then the news channel is small.

Using the conversion factor together with Assumption \ref{ass:bounded_news}, we arrive at the following bounds for the MPC with respect to changes in current cash-on-hand. 

\begin{proposition}[Bounds on the MPC]
\label{prop:mpc_bounds}
Under Assumptions \ref{ass:income_process}--\ref{ass:bounded_news}, the MPC with respect to current cash-on-hand satisfies
\begin{equation}
    \frac{1}{1+\sum_{j=1}^{k}\theta_j}
    \frac{C_{i,t}}{Y_{i,t}}
    \gamma(m_{i,t-1},\mathbf{s}_{i,t-1})
    \le
    MPC_{i,t}
    \le
    \frac{C_{i,t}}{Y_{i,t}}
    \gamma(m_{i,t-1},\mathbf{s}_{i,t-1}).
    \label{eq:mpc_bounds}
\end{equation}
\end{proposition}

\noindent\textit{Proof:} See Appendix Section~\ref{app:identification_proofs}.

Proposition \ref{prop:passthrough} and Proposition \ref{prop:mpc_bounds} guide the empirical analysis. The empirical section estimates $\gamma$ as a function of the predetermined state and then converts it into MPC bounds using \eqref{eq:mpc_bounds}. In our application, the estimated persistence of transitory income is modest, so the lower and upper bounds are tight.

\section{Estimation}
\label{sec:estimator}

To recover the cash-on-hand gradient in the MPC, we need to estimate the state-dependent pass-through coefficient $\gamma$ in Equation \eqref{eq:passthrough_equation}. This section describes how we estimate it. 

The key premise of our approach is that income is measured with little error. This premise is justified in administrative data, where the main income components are tax-based and third-party reported. It is less suitable for survey data, for which canonical estimators were originally developed, such as BPP and the refined BPP-C. Those estimators sacrifice precision in order to remain robust to income measurement error. We instead use the full income history to recover the latent shocks as precisely as possible. 

Our estimation procedure has three steps. First, we estimate the common income process from pooled income-growth moments. Second, we recover household-level permanent shocks $\widehat{\eta}_{i,t}$ and transitory shocks $\widehat{\varepsilon}_{i,t}$ using the Kalman smoother based on household income histories. Third, we estimate the empirical counterpart of the theoretical pass-through equation using the recovered shocks, within cells of the predetermined state variables:
\begin{equation}
    \Delta \log C_{i,t}
    =
    \mu
    +
    b_{\eta}(\cdot)\widehat{\eta}_{i,t}
    +
    b_{\varepsilon}(\cdot)\widehat{\varepsilon}_{i,t}
    +
    u_{i,t}.
    \label{eq:second_stage_estimation}
\end{equation}

Equation \eqref{eq:second_stage_estimation} is written with generic population regression coefficients $b_{\eta}(\cdot)$ and $b_{\varepsilon}(\cdot)$. The residual $u_{i,t}$ captures variation in consumption behavior orthogonal to the income-history signal, such as taste shocks and independent measurement error in consumption. 

There is no mechanical reason why the coefficient $b_{\varepsilon}(\cdot)$ in the estimating equation \eqref{eq:second_stage_estimation} needs to inform us about the theoretical pass-through coefficient $\gamma(\cdot)$ in Equation \eqref{eq:passthrough_equation}. Here, the regressors are recovered signals rather than the latent shocks themselves. Using the fact that the Kalman smoother provides the best linear projection of the latent shocks, we prove identification and efficiency under the assumption of no measurement error in income or the underlying state variables. We then show that the bias resulting from relevant magnitudes of measurement errors in income, or estimated state variables, are small.

\subsection{Income Process Parameters}
The first step of our procedure is to estimate the income process parameters. We estimate $(\theta_1,\ldots,\theta_k,\sigma^2_\varepsilon,\sigma^2_\eta)$ from pooled autocovariances of residual income growth, following standard practice. An MA$(k)$ transitory component in income levels implies an MA$(k+1)$ component in income growth:
\begin{equation}
    \Delta \log Y_{i,t}
    =
    \eta_{i,t}
    +
    \sum_{j=0}^{k+1}\psi_j\varepsilon_{i,t-j},
    \label{eq:income_growth_ma}
\end{equation}
where $\psi_0=1$, $\psi_j=\theta_j-\theta_{j-1}$ for $1\leq j\leq k$, and $\psi_{k+1}=-\theta_k$. Let $a_{\ell} \equiv \operatorname{Cov}(\Delta\log Y_{i,t+\ell},\Delta\log Y_{i,t})$ denote the autocovariance of income growth at lag $\ell$. The autocovariances satisfy
\begin{align}
    a_0
    &=
    \sigma_\eta^2
    +
    \sigma_\varepsilon^2\sum_{j=0}^{k+1}\psi_j^2,
    \label{eq:income_autocov_variance}\\
    a_{\ell}
    &=
    \sigma_\varepsilon^2
    \sum_{j=\ell}^{k+1}\psi_j\psi_{j-\ell},
    \qquad
    \ell=1,\ldots,k+1.
    \label{eq:income_autocov_lags}
\end{align}
We choose the values of $(\theta_1,\ldots,\theta_k,\sigma^2_\varepsilon)$ that match the nonzero autocovariances $a_1,\ldots,a_{k+1}$, and then recover $\sigma_\eta^2$ from the variance identity \eqref{eq:income_autocov_variance}.

\subsection{Kalman smoothing}

The second step is to recover the latent shocks. Conditional on the income process parameters, we write \eqref{eq:income_growth_ma} as a standard state-space model with state vector
\begin{equation}
    \boldsymbol{\alpha}_{i,t}
    =
    (\eta_{i,t},\varepsilon_{i,t},\varepsilon_{i,t-1},\ldots,\varepsilon_{i,t-k-1})'.
    \label{eq:state_vector_estimation}
\end{equation}
and denote each household income history with
\[
    \mathbf{y}_i
    =
    (\Delta\log Y_{i,1},\ldots,\Delta\log Y_{i,T_i})'.
\]

We run a Kalman filter followed by a fixed-interval smoother. Appendix~\ref{app:kalman_smoother} gives the state-space implementation details. The filter forms one-sided predictions using income growth up to date $t$. The smoother uses the full income path, including future income growth, to decompose observed income changes into permanent and transitory shocks.

The recovered shocks form the best linear projection of the latent state on $\mathbf{y}_i$. By best linear projection, we mean that $\operatorname{Proj}(X\mid \mathbf{Z})=\operatorname{Cov}(X,\mathbf{Z})
\operatorname{Var}(\mathbf{Z})^{-1}\mathbf{Z}$, given that $X$ and $\mathbf{Z}$ have been centered. One may therefore interpret the recovered shocks as the permanent and current-transitory components of this projection:
\begin{equation}
    (\widehat{\eta}_{i,t},\widehat{\varepsilon}_{i,t})
    =
    \operatorname{Proj}\left((\eta_{i,t},\varepsilon_{i,t})
    \mid
    \mathbf{y}_i\right).
    \label{eq:smoothed_shocks}
\end{equation}
This interpretation is valid for any distribution of the latent shocks with finite variances and serially uncorrelated innovations. This is reassuring, because household income changes are in general not well described by a normal distribution, as highlighted by the recent literature on income process estimation, see, e.g., \cite{Guvenen2021}. Our estimator only uses the covariance structure of the income process, and not the full distribution of income shocks. With Gaussian shocks, however, the interpretation of the recovered shocks strengthens: the recovered shocks are in that case conditional expectations, not merely best linear predictions.

Figure \ref{fig:transitory_permanent_footprints} illustrates the logic of the Kalman smoother, interpreted as a model-based event-study filter. Around each date, the income path is compared to the dynamic footprints of permanent and transitory shocks. A permanent shock leaves a persistent change in income. A transitory shock leaves the MA footprint (a ``spike'') implied by the estimated $\theta$'s. The smoother assigns the date-$t$ shock the best linear prediction implied by these footprints, using all available dates rather than a single future-income moment. The same observed income movement can be decomposed differently depending on what happens before and after the date of interest, and thus both past and future income growth are essential to forming the prediction. 

\begin{figure}[t]
    \centering
    \caption{Dynamic footprints of income shocks}
    \label{fig:transitory_permanent_footprints}
    \includegraphics[width=0.86\textwidth]{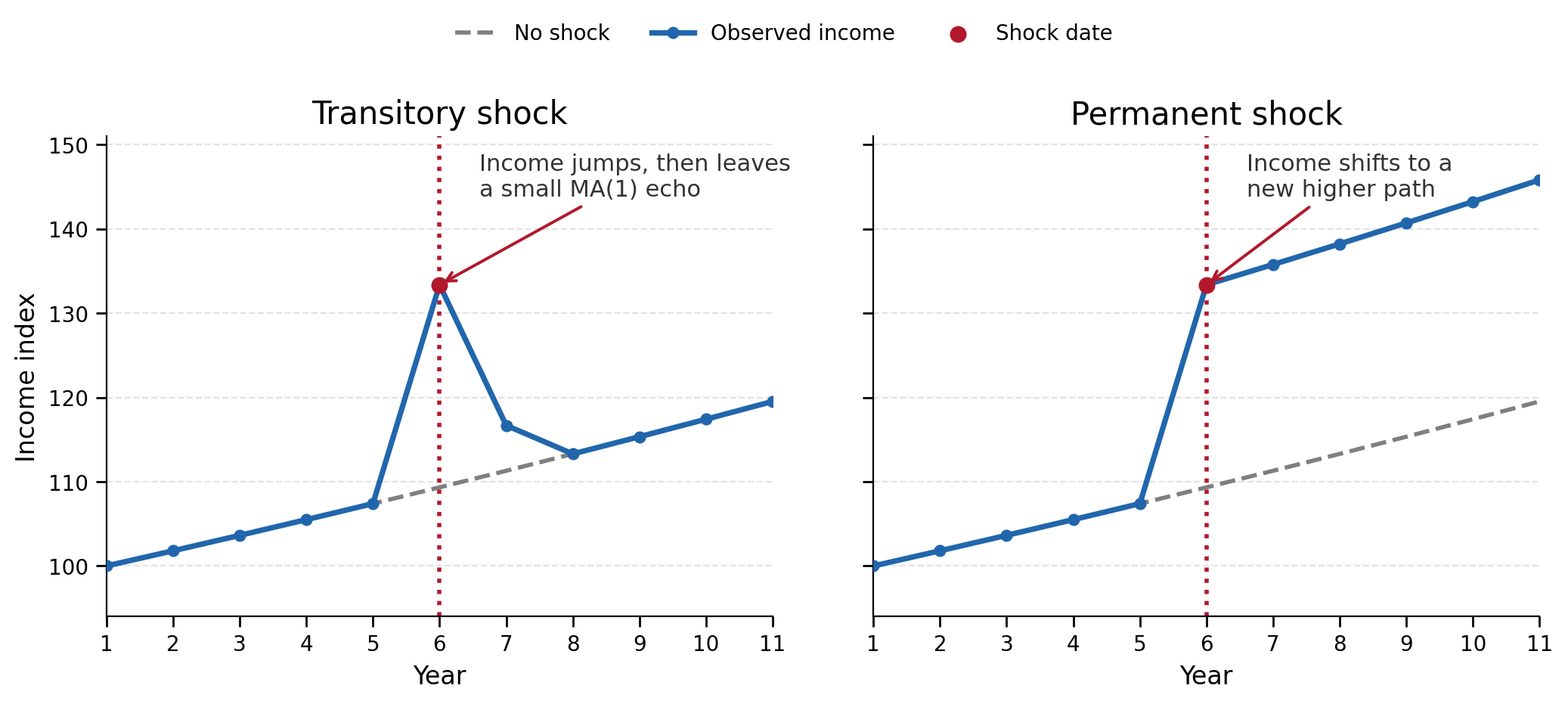}
    \figurenote{The transitory shock is illustrated with a small MA(1) component, calibrated to the persistence estimated in our baseline income process.}
\end{figure}

\subsection{Identification}

The recovered shock is a generated regressor. The key projection identity in Equation \eqref{eq:smoothed_shocks} is nevertheless enough for the population coefficient in \eqref{eq:second_stage_estimation} to identify the pass-through coefficient $\gamma$, provided that the pass-through equation is correctly specified after conditioning on the true state, the permanent shock is controlled for, and the residual consumption-growth component is orthogonal to the income-history signal.

\begin{proposition}[Identification of the pass-through coefficient]
\label{prop:passthrough_identification}
Let the state-dependent pass-through equation hold within a state cell, or after partialling out state controls, and suppose that income is measured without error. Suppose also that the residual component $\omega_{i,t}$ is orthogonal to the income-history signal after conditioning on the state and the permanent-shock component. Let
\[
    \widehat{\varepsilon}_{i,t}
    =
    \operatorname{Proj}(\varepsilon_{i,t}\mid \mathbf{y}_i),
    \qquad
    \mathbf{y}_i=(\Delta\log Y_{i,1},\ldots,\Delta\log Y_{i,T_i})'.
\]
Then the population coefficient $b_{\varepsilon}(\cdot)$ in Equation \eqref{eq:second_stage_estimation} equals $\gamma$ in Equation \eqref{eq:passthrough_equation}.
\end{proposition}

\noindent\textit{Proof:} See Appendix Section~\ref{app:minvar_signal}.

The intuition is that the shock obtained from the Kalman smoother is a noisy signal of the true transitory shock, but the noise is the projection error. By construction, this projection error is orthogonal to the signal itself. Therefore the covariance between the signal and the true shock equals the variance of the signal. Once the permanent shock and the state variables have been controlled for, this is the covariance restriction needed for the regression coefficient on $\widehat{\varepsilon}_{i,t}$ to recover $\gamma$.

The permanent-shock control is essential. If it is omitted, the probability limit of the transitory coefficient contains the omitted-variable term
\begin{equation}
    \gamma
    +
    \lambda
    \frac{\operatorname{Cov}(\widehat{\varepsilon}_{i,t},\eta_{i,t})}
         {\operatorname{Var}(\widehat{\varepsilon}_{i,t})}.
    \label{eq:omitted_permanent_bias}
\end{equation}
The point is not that the underlying shocks are correlated. By assumption, $\eta_{i,t}$ and $\varepsilon_{i,t}$ are orthogonal. The issue is that the estimated signals are both constructed from the same observed income history. The smoothed transitory signal can therefore be correlated with the permanent shock component, even when the latent innovations are orthogonal. Thus the estimator is not a regression of consumption growth on the transitory signal alone. It is a regression of consumption growth on the transitory signal after accounting for the permanent shock and the state variables that define the pass-through coefficient.

\subsection{Efficiency}

The same projection logic also gives an efficiency result. Consider estimators that consistently identify $\gamma$ using signals or instruments that are linear functions of the observed income history, which includes our projection estimator. Within this class, our estimator is the most efficient one.

\begin{proposition}[Minimum-variance linear income-history signal]
\label{prop:estimation_minvar}
Among all valid scalar instruments $Z_{i,t}=\mathbf{a}'\mathbf{y}_i$ that are linear in the income history and identify $\gamma$, the projection signal
\[
    \widehat{\varepsilon}_{i,t}
    =
    \operatorname{Proj}(\varepsilon_{i,t}\mid \mathbf{y}_i)
\]
minimizes the asymptotic variance of the IV estimator of $\gamma$. For any admissible $Z_{i,t}$,
\begin{equation}
    \frac{\operatorname{Avar}(\widehat{\gamma}(Z))}
         {\operatorname{Avar}(\widehat{\gamma}(\widehat{\varepsilon}))}
    =
    \frac{\rho^2(\widehat{\varepsilon}_{i,t},\varepsilon_{i,t})}
         {\rho^2(Z_{i,t},\varepsilon_{i,t})}
    \geq 1.
    \label{eq:minvar_ratio}
\end{equation}
Equality holds if and only if $Z_{i,t}$ is proportional to $\widehat{\varepsilon}_{i,t}$.
\end{proposition}

\noindent\textit{Proof:} See Appendix Section~\ref{app:minvar_signal}.

The intuition is simple. Precision is governed by how strongly the identifying signal is correlated with the latent transitory shock. The best linear projection is, by definition, the linear combination of the income history that maximizes this correlation, which is the Kalman-smoothed shock. All other valid linear income-history signals use less information, or use the same information with non-optimal weights.

The same logic underpinning our identification and efficiency results for the transitory pass-through coefficient $\gamma$ also applies to the permanent-shock pass-through coefficient $\lambda$. The best linear projection of the permanent shock on the income history is the Kalman-smoothed permanent shock, and all other valid linear income-history signals use less information, or use the same information with non-optimal weights. We focus on the transitory pass-through coefficient in the main text because it is the key object of interest for the MPC gradient.

\subsection{Relation to BPP and BPP-C}

The class of estimators that consistently identify $\gamma$ using linear income-history signals includes most estimators in the literature. In particular, it includes the BPP estimator and the refined BPP-C estimator. The BPP-C estimator is adopted in recent studies using administrative data such as \citet{FlorinOBilbiie2025} and \citet{Ganong2025}.

BPP estimates pass-through from covariance restrictions implied by the permanent-transitory income process. The key idea is to use future income growth to isolate the transitory component of current income growth. Future income growth is orthogonal to the current permanent shock, because a permanent shock changes the level of future income but not future income growth. At the same time, future income growth is correlated with the current transitory shock because the transitory component leaves a finite dynamic footprint in income growth.

In our notation, the BPP instruments are selected elements of the income history,
\[
    \mathbf{Z}^{BPP}_{i,t}
    =
    R_t\mathbf{y}_i
    =
    (\Delta\log Y_{i,t+1},\Delta\log Y_{i,t+2},\ldots,\Delta\log Y_{i,t+q})',
    \qquad
    \mathbf{y}_i=(\Delta\log Y_{i,1},\ldots,\Delta\log Y_{i,T_i})'.
\]
Here $R_t$ is a selection matrix. The moment condition for the transitory pass-through can be written as
\[
    \mathbb{E}\left[
    \mathbf{Z}^{BPP}_{i,t}
    \left(
    \widetilde{\Delta\log C}_{i,t}
    -
    \gamma\varepsilon_{i,t}
    \right)
    \right]
    =
    0,
\]
or equivalently as a set of covariance restrictions using the model-implied vector $\operatorname{Cov}(\mathbf{Z}^{BPP}_{i,t},\varepsilon_{i,t})$, which is why BPP fits into our comparison class. Although their estimator is written as a GMM, every moment is built from a linear function of the income history. With optimal GMM weighting, the estimator uses an effective scalar signal $\mathbf{w}'\mathbf{Z}^{BPP}_{i,t}=\mathbf{w}'R_t\mathbf{y}_i$, which is still linear in $\mathbf{y}_i$.

The virtue of the BPP moments is robustness. By relying on future income growth rather than contemporaneous income growth, the estimator avoids the direct contamination that classical income measurement error creates at date $t$. The cost is that the selected future leads need not be the most informative linear combination of the household's income path. The Kalman-smoothed shock uses the same income process, but combines all dates at which a transitory shock leaves a footprint.

The BPP-C estimator starts from the same logic but asks which future-income moment remains valid when log consumption growth is allowed to also respond (linearly) to past transitory shocks. This is particularly relevant in settings where one cannot condition on the relevant state variables, which is typically the case in small-sample survey data such as the PSID. Without state-dependence, Equation \eqref{eq:passthrough_equation} implies that consumption is a random walk, which is not consistent with standard theory.

If past transitory shocks also affect current consumption growth, some future-income moments are correlated with both the current transitory shock and the past shocks entering consumption growth. The BPP-C estimator therefore selects the future lead that filters out both the permanent shock and the past transitory shocks. With an MA$(k)$ transitory income component, this robust moment uses
\[
    Z^C_{i,t}
    =
    \Delta\log Y_{i,t+k+1}.
\]
This instrument is a single element of the same income-history vector, $Z^C_{i,t}=\boldsymbol{\iota}_{t+k+1}'\mathbf{y}_i$, and is therefore also linear in income history.

Relative to BPP, the virtue of the BPP-C estimator is thus not only robustness to measurement error in income, but also robustness to misspecification of the consumption function. The price of BPP-C's robustness is an even weaker signal: it uses only one future-income realization, namely the most distant realization that still contains information. The cost is that this distant lead can be a weak signal: its identifying covariance is governed by the last MA coefficient, which is small when transitory-income persistence is modest. BPP-C can therefore suffer from a weak-instrument problem, as we illustrate in Section~\ref{sec:empirical_results}.

In contrast, our projection estimator uses the full dynamic footprint of the transitory shock and is more efficient. Moreover, our approach addresses \citet{Commault2022}'s misspecification critique of BPP by allowing for state-dependence. If there are other behavioral channels through which past transitory shocks affect current consumption growth not encapsulated by Assumption~\ref{ass:consumption_behavior}, our approach remains valid as long as these past shocks are also controlled for when estimating Equation \eqref{eq:second_stage_estimation}.\footnote{Maintaining the assumptions of no measurement error in income and a correctly specified state-conditioned consumption equation, it is possible to use near-lead income growth as an instrument together with ex-post bias correction to recover the true pass-through coefficient. We elaborate on this in Appendix Section~\ref{app:bpp_c_near_lead_bias_correction}.}

\subsection{Income measurement error implies a (favorable) bias-variance trade-off}

Propositions \ref{prop:passthrough_identification} and \ref{prop:estimation_minvar} give identification and efficiency assuming no measurement error in income. The administrative data used here likely contains very limited measurement error, but we cannot guarantee it is zero. It may also be interesting to ask what happens when implementing our estimator in settings with substantial income measurement error, such as the PSID and other survey data.

With measurement error in income, the Kalman-smoothed shock is computed from a contaminated income history and the estimator can be biased. Intuitively, with i.i.d. measurement error, the Kalman smoother will not be able to distinguish noise from short-lived transitory income shocks; both would look like spikes in the income growth history, as in the left panel of Figure \ref{fig:transitory_permanent_footprints}. The resulting shock series thus contains noise, and the projection estimator will feature attenuation bias, in contrast to the robust BPP and BPP-C estimators. Nevertheless, the estimator is still very efficient as it uses more information to recover the shocks. 

Appendix \ref{app:measurement_error_simulations} describes two simulation exercises that quantify this trade-off. The first mimics survey data, where income and consumption measurement errors are i.i.d. and independent. In that setting, the BPP-C estimator remains centered on the true pass-through coefficient, while the naive linear-projection estimator is attenuated as income measurement error rises. The second mimics budget-constraint imputation (as in our data), where income measurement error mechanically enters imputed consumption. In that setting, the same error appears in both the income and consumption equations, and the naive linear-projection estimator can be biased upward. Across both simulations, the bias is very small and the confidence intervals for the BPP-C estimator are much larger than for the ``naive'' linear-projection estimator for all reasonable levels of income measurement error. 

\subsection{Measurement error in the state variables implies (small) attenuation bias}
\label{subsec:state_measurement_error}

The preceding identification argument treats the predetermined state variables as observed. In the empirical implementation, however, the state is partly estimated. This matters because the cash-on-hand gradient in the MPC is identified from differences in pass-through coefficients across state cells. Per proposition \ref{prop:passthrough_identification}, noise stemming from the estimation in a state variable does not generally bias the average pass-through coefficient, but it mixes households across true state cells. If the true pass-through coefficient is decreasing in the state, observations assigned to a low-state cell include some households whose true state is higher, and observations assigned to a high-state cell include some households whose true state is lower. This attenuates the estimated gradient toward zero, and that estimated gradient should therefore be interpreted as a conservative lower bound of the true gradient's magnitude.

In this paper, we are primarily interested in how the pass-through coefficient and its associated MPC vary with one of the state variables: normalized cash-on-hand. For normalized cash-on-hand, the measurement-error problem has a useful special structure. Recall that normalized cash-on-hand is defined as $m_{i,t}=\frac{M_{i,t}}{P_{i,t}}$, where $M_{i,t}$ is cash-on-hand and $P_{i,t}$ is permanent income. In our data, the liquid-wealth component of $M_{i,t}$ is observed in the registers, while $P_{i,t}$ is recovered from the estimated income process. After residualizing income (removing $Y^d$), we can write log income as $\log Y_{i,t}=\log P_{i,t}+\nu_{i,t}$, where $\nu_{i,t}$ is the transitory income component. The smoothed permanent-income estimate can similarly be written as $\log \widehat P_{i,t}=\log Y_{i,t}-\widehat\nu_{i,t}$. It follows that
\[
    \log \widehat m_{i,t}-\log m_{i,t}
    =\widehat\nu_{i,t}-\nu_{i,t}.
\]
Thus, the sorting error in normalized cash-on-hand is the recovery error in the transitory income component. 

This identity bounds the magnitude of the problem. Since the Kalman-smoothed transitory component is the linear projection of $\nu_{i,t}$ on the income history, the projection error is orthogonal to the recovered component. In particular, 
\[
    \operatorname{Var}(\widehat\nu_{i,t}-\nu_{i,t})
    =
    \operatorname{Var}(\nu_{i,t})(1-R^2_\nu),
\]
where $R^2_\nu \in [0,1]$ is the reliability of the recovered transitory component. Under the baseline MA(1) specification, $\nu_{i,t}=\varepsilon_{i,t}+\theta_1\varepsilon_{i,t-1}$. The Singles estimates in Table~\ref{tab:income_process_moments} imply $\operatorname{sd}(\nu_{i,t}) \simeq 0.11$; the corresponding number for Households is also about 0.11. This is an upper bound on the log cash-on-hand sorting error before using the information in the household's income history. To interpret the scale, the empirical standard deviation of log normalized cash-on-hand is about 0.32 in both samples, which means that worst case attenuation of a linear slope under classical measurement error is about 11 percent.\footnote{A simple classical-error calculation gives the interpretation. If $\log \widehat m_{i,t}=\log m_{i,t}+e_{i,t}$ with $e_{i,t}$ independent of $\log m_{i,t}$ and $\operatorname{sd}(e_{i,t})\leq 0.11$, then
\[
    \operatorname{Corr}(\log \widehat m_{i,t},\log m_{i,t})
    \geq
    \sqrt{
    \frac{\operatorname{Var}(\log m_{i,t})}
    {\operatorname{Var}(\log m_{i,t})+0.11^2}
    }.
\]
For a linear pass-through gradient in log cash-on-hand, the corresponding upper bound on slope attenuation is
\[
    \frac{\operatorname{Var}(\log m_{i,t})}
    {\operatorname{Var}(\log m_{i,t})+0.11^2}.
\]
These formulas are illustrative because the empirical estimator bins observations into deciles and the Kalman recovery error need not be exactly classical.} The Kalman smoother reduces this error further in proportion to the reliability of the recovered transitory component, meaning that the actual attenuation is likely much smaller.

\section{Data and Measurement}
\label{sec:data}

\subsection{Data Description}

Our empirical analysis uses yearly observations from Swedish administrative registers linked at the individual and household level over 2000--2007. The main source is LISA, which merges tax and transfer registers for the universe of individuals aged 16 and above and provides annual information on labor earnings, capital income, transfers, taxes, and rich socio-demographic characteristics. We link LISA to complementary wealth registers covering financial and real assets. We keep the exposition concise here, and provide all relevant details in Appendix~\ref{app:data}.

We impute consumption from income and wealth using the household budget constraint,
\[
C_{i,t}=Y_{i,t}-S_{i,t},
\]
where $Y_{i,t}$ is disposable income and $S_{i,t}$ captures active saving flows (net asset purchases and net debt repayment). Savings are constructed asset by asset using quantity- and price-based approaches so that passive valuation changes do not mechanically enter consumption, following \cite{Kolsrud_et_al_2020_JPUB} and \cite{BrowningLethPetersen2003}. Appendix~\ref{app:data} reports all details in this procedure and compares the resulting aggregate series to national accounts and survey consumption.

Our measure of disposable income includes all non-capital income, net of taxes and transfers. We do not include capital income, which is endogenous to consumption and saving decisions. We residualize log income and log consumption before constructing growth rates using standard characteristics following \citet{Blundell_et_al_2008_AER}, including time fixed effects.

Wealth is measured in terms of the financial position of households at the end of each calendar year. The wealth register reports only position above 10,000 SEK ($\sim$ 1200 USD). Observations below 10,000 SEK are thus treated as zeros. We measure wealth in three components. \emph{Liquid wealth} ($A^{\mathrm{liq}}_{i,t}$) is the sum of all financial assets, including bank deposits, listed securities, and other marketable financial positions. It does not include the pension claims within the public pension system, as these claims are not accessible prior to retirement. \emph{Net illiquid wealth} ($A^{\mathrm{illiq}}_{i,t}$) is housing and other non-marketable real assets, primarily owner-occupied and co-operative dwellings valued at register prices, net of all debt (which mainly consists of mortgages and government-provided student debt). \emph{Net total wealth} ($A^{\mathrm{net}}_{i,t}$) sums liquid wealth and net illiquid wealth.

We work with two estimation samples. \emph{Singles} is our baseline sample: individuals aged 30--65 who are not married or cohabiting and who satisfy basic data-quality restrictions. \emph{Households} adds married and cohabiting couples observed in the same dwelling. For households, we estimate the income processes on the joint income of the household members, after residualizing with respect to characteristics of the household head. We selected Singles as our baseline sample to avoid any issues with income process estimation that result from residualization and within-household insurance. 

When estimating the second stage consumption function equation \eqref{eq:second_stage_estimation}, both samples exclude individual-year observations with transactions in illiquid assets, that is, housing and co-operative-apartment purchases and sales. We do this because we assume a single liquid asset when deriving our estimating equation. From the perspective of a model that also includes illiquid assets, as in \citet{Kaplan_Violante_2014_Econometrica} and \citet{Bayer_et_al_2019_Econometrica}, our pass-through coefficients should be interpreted as conditional on not adjusting the illiquid account. 

\subsection{Descriptive Statistics}
\label{sec:descriptive_stats}

Table~\ref{tab:summary_stats} reports medians and means for the Households and Singles estimation samples. Monetary values are expressed in 2003 SEK, when 1 SEK $\sim$ 0.12 USD. Average consumption tracks disposable income closely in both samples, which is in line with the very low Swedish aggregate savings rate. Specifically, the aggregate savings rate as a fraction of disposable income was 2\% on average during 1999--2007.\footnote{Note that the aggregate savings rate, unlike our microdata sample, includes the dissaving of retired individuals.} For income, consumption, and wealth, medians are significantly lower than means, reflecting the fat right tails of all distributions.

\begin{table}[t]
    \centering
    \caption{Summary statistics}
    \label{tab:summary_stats}  
    \begin{tabular}{@{}lrrrr@{}}
\toprule
 & \multicolumn{2}{c}{Households} & \multicolumn{2}{c}{Singles} \\
\cmidrule(lr){2-3} \cmidrule(lr){4-5}
 & Median & Mean & Median & Mean \\
\midrule
Share of males (\%) & 48.28 & 48.28 & 57.31 & 57.31 \\
Male age & 44.00 & 45.42 & 45.00 & 45.83 \\
Female age & 48.00 & 47.71 & 52.00 & 49.51 \\
Household size & 1.00 & 1.68 & 1.00 & 1.00 \\
\addlinespace
Disposable income & 188,357 & 215,397 & 164,131 & 171,506 \\
Cash on hand & 224,643 & 1,441,242 & 190,357 & 273,404 \\
Consumption & 192,074 & 244,084 & 163,389 & 192,607 \\
\addlinespace
Net total wealth & 20,424 & 1,499,607 & 8,090 & 317,870 \\
Liquid wealth & 12,536 & 1,278,938 & 8,294 & 148,760 \\
Illiquid wealth & 0 & 173,537 & 0 & 127,667 \\
\bottomrule
\end{tabular}
\par\smallskip
\begin{minipage}{0.82\textwidth}
\footnotesize
\emph{Notes:} Monetary values are in 2003 SEK. Samples pool household-year observations over 2000--2007 (Households: 8,037,370; Singles: 5,631,963). Means are sensitive to top wealth observations.
\end{minipage}

\end{table}

Figure~\ref{fig:wealth_distr} plots the distributions of net total wealth, net illiquid wealth, and liquid wealth. Wealth is highly skewed in both samples: medians are modest in Table~\ref{tab:summary_stats}, while the right tails are long, especially for Households. The panels for net illiquid wealth cap the vertical axis to make the tail visible. A large share of observations have zero measured net illiquid wealth (recall the cap from below at 10,000 SEK), reflecting the discrete nature of home ownership. Figure~\ref{fig:cons_disp_inc_distr} plots the distributions of consumption, disposable income, and lagged cash-on-hand.

\begin{figure}[p]
    \centering
    \captionsetup[subfigure]{font=footnotesize,skip=1pt}
    \caption{Distribution of wealth components}
    \label{fig:wealth_distr}
    \begin{subfigure}[t]{0.44\textwidth}
        \centering
        \includegraphics[width=\linewidth]{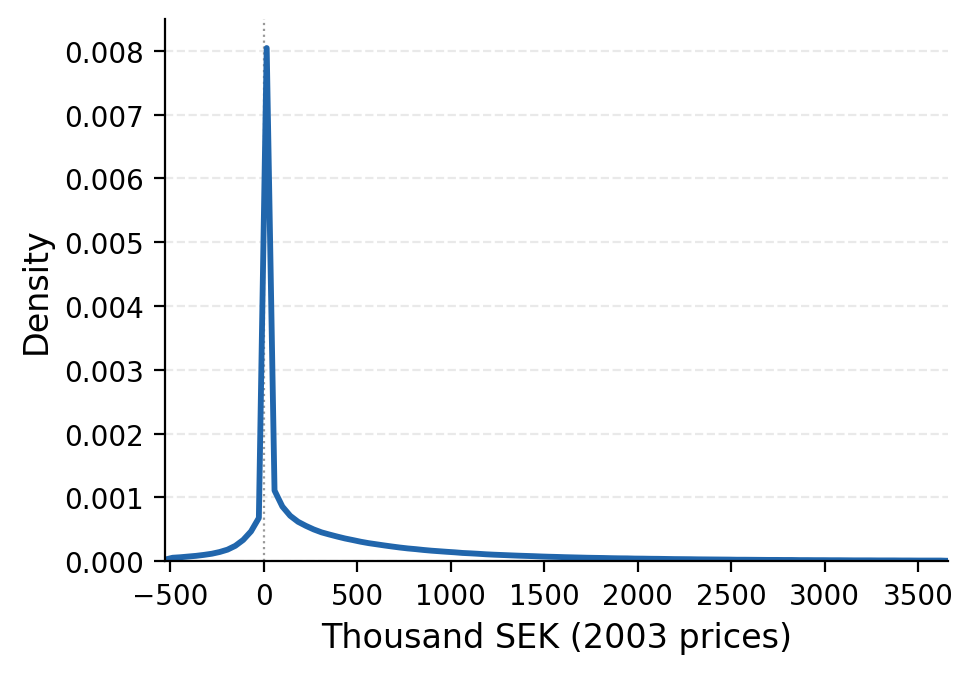}
        \caption{Net total wealth, households}
    \end{subfigure}\quad
    \begin{subfigure}[t]{0.44\textwidth}
        \centering
        \includegraphics[width=\linewidth]{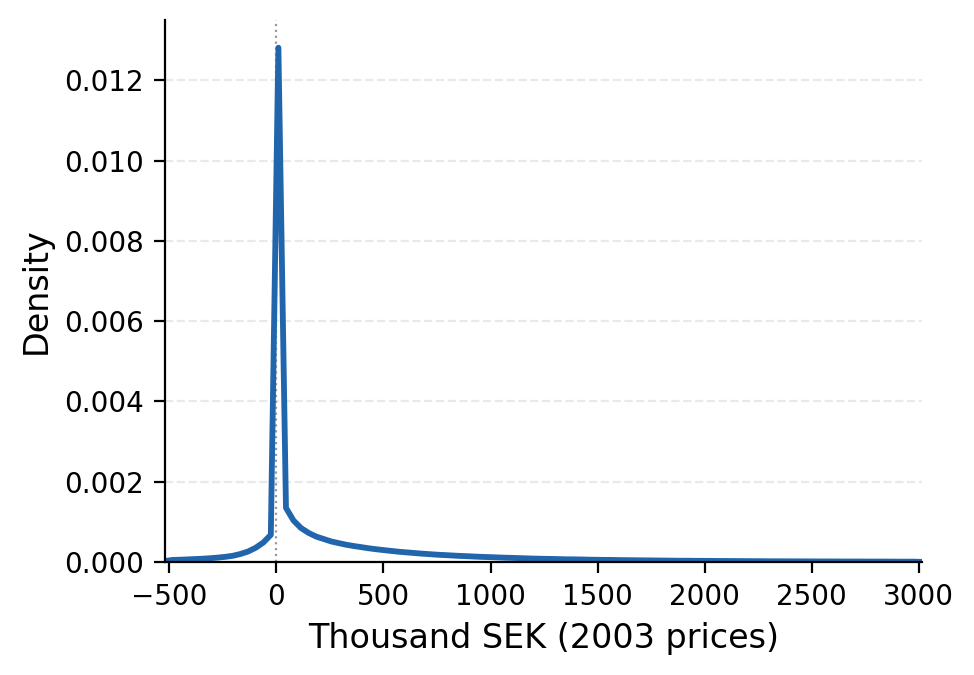}
        \caption{Net total wealth, singles}
    \end{subfigure}

    \par\medskip

    \begin{subfigure}[t]{0.44\textwidth}
        \centering
        \includegraphics[width=\linewidth]{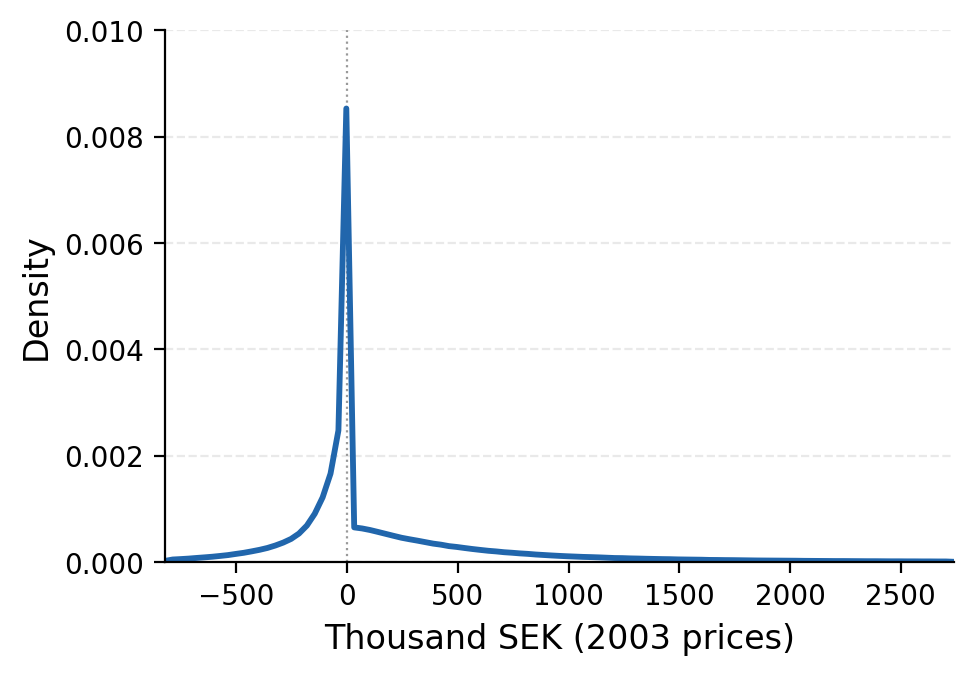}
        \caption{Net illiquid wealth, households}
    \end{subfigure}\quad
    \begin{subfigure}[t]{0.44\textwidth}
        \centering
        \includegraphics[width=\linewidth]{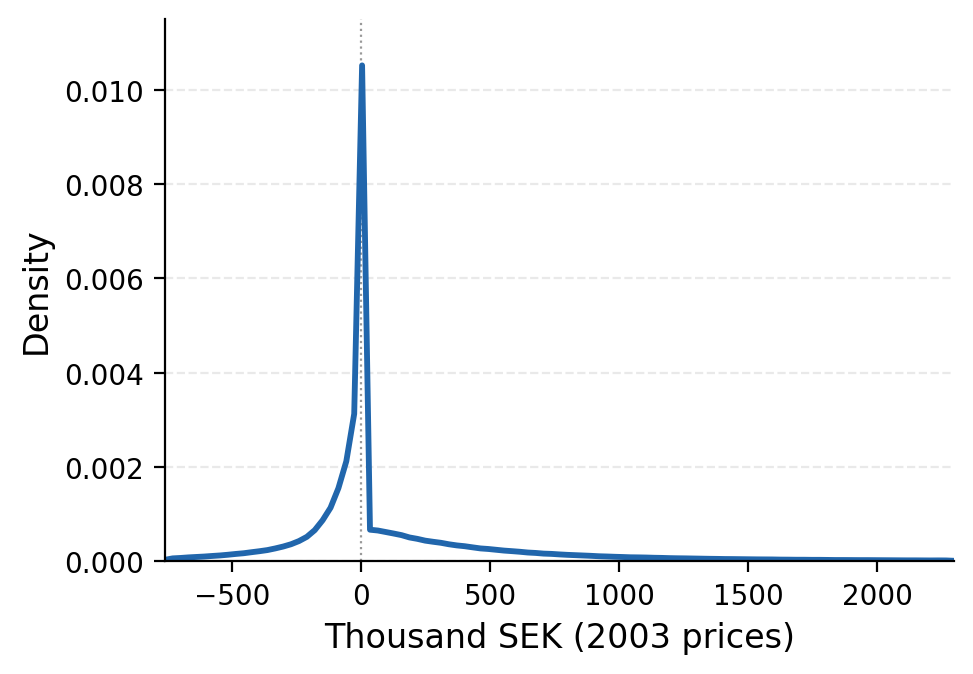}
        \caption{Net illiquid wealth, singles}
    \end{subfigure}

    \par\medskip

    \begin{subfigure}[t]{0.44\textwidth}
        \centering
        \includegraphics[width=\linewidth]{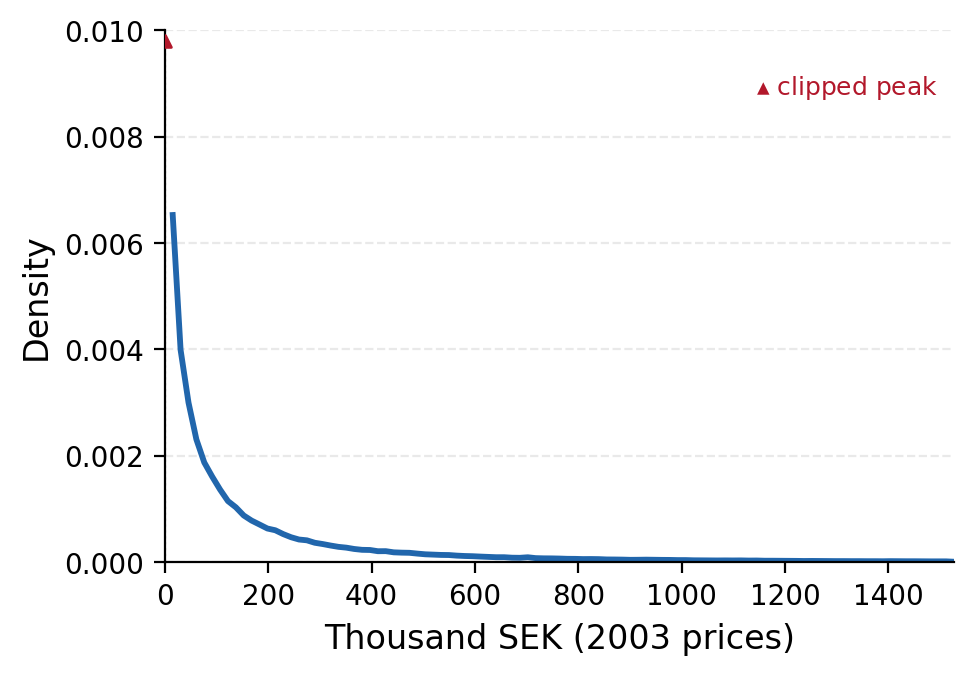}
        \caption{Liquid wealth, households}
    \end{subfigure}\quad
    \begin{subfigure}[t]{0.44\textwidth}
        \centering
        \includegraphics[width=\linewidth]{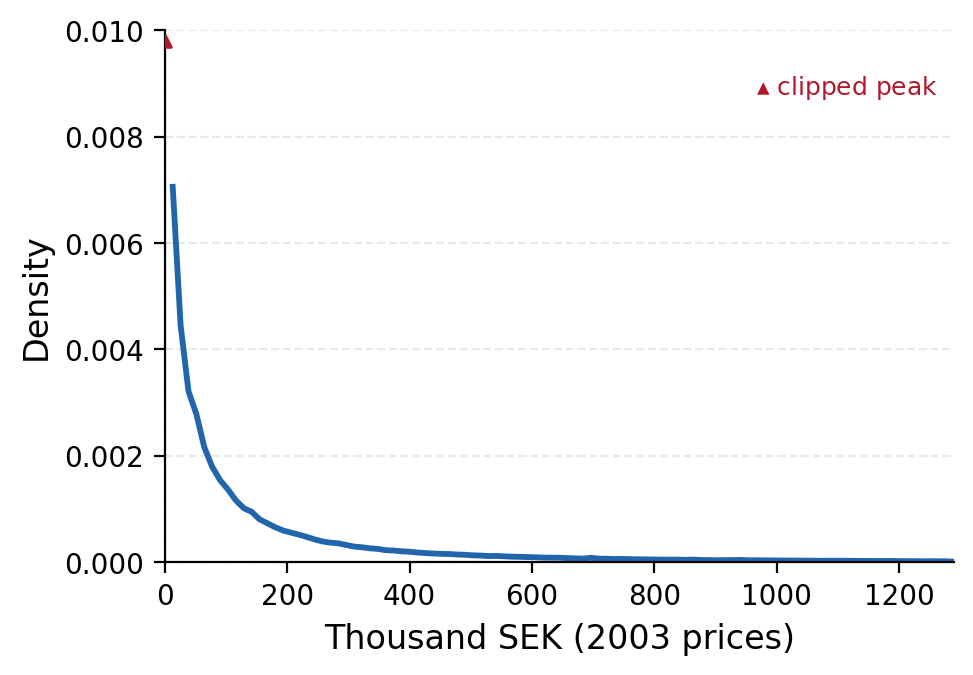}
        \caption{Liquid wealth, singles}
    \end{subfigure}
    \figurenote{Kernel density estimates pool household-year observations over 2000--2007. Variables are in thousand 2003 SEK. Peaks in the liquid-wealth panels are capped to show tail behavior.}
\end{figure}
\clearpage

\begin{figure}[p]
    \centering
    \caption{Distribution of consumption, disposable income, and cash-on-hand}
    \label{fig:cons_disp_inc_distr}
    \begin{subfigure}[t]{0.44\textwidth}
        \centering
        \includegraphics[width=\linewidth]{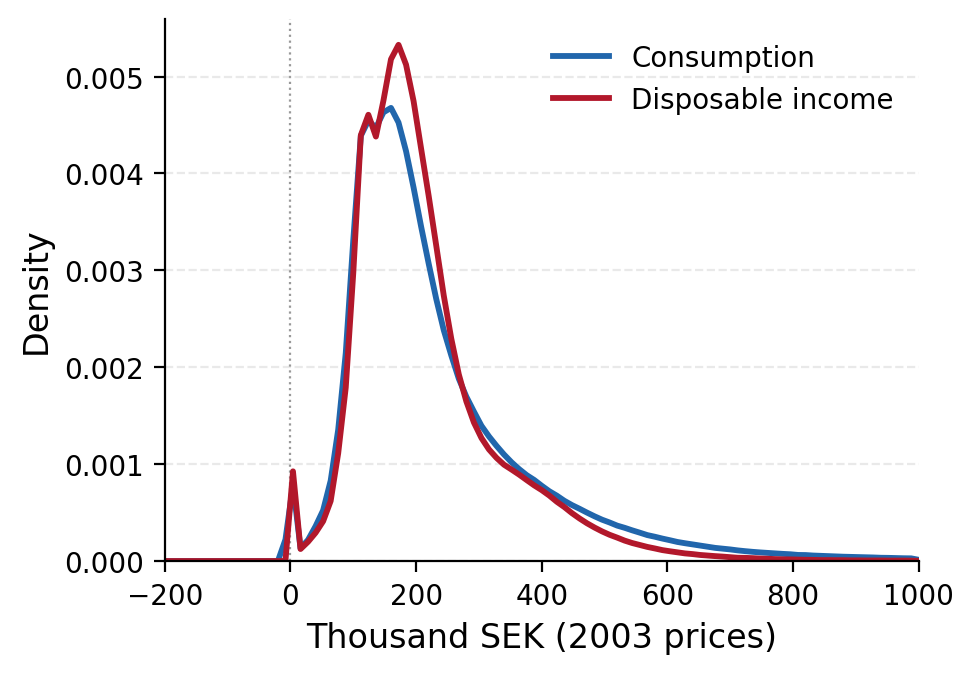}
        \caption{Consumption and income, households}
    \end{subfigure}\quad
    \begin{subfigure}[t]{0.44\textwidth}
        \centering
        \includegraphics[width=\linewidth]{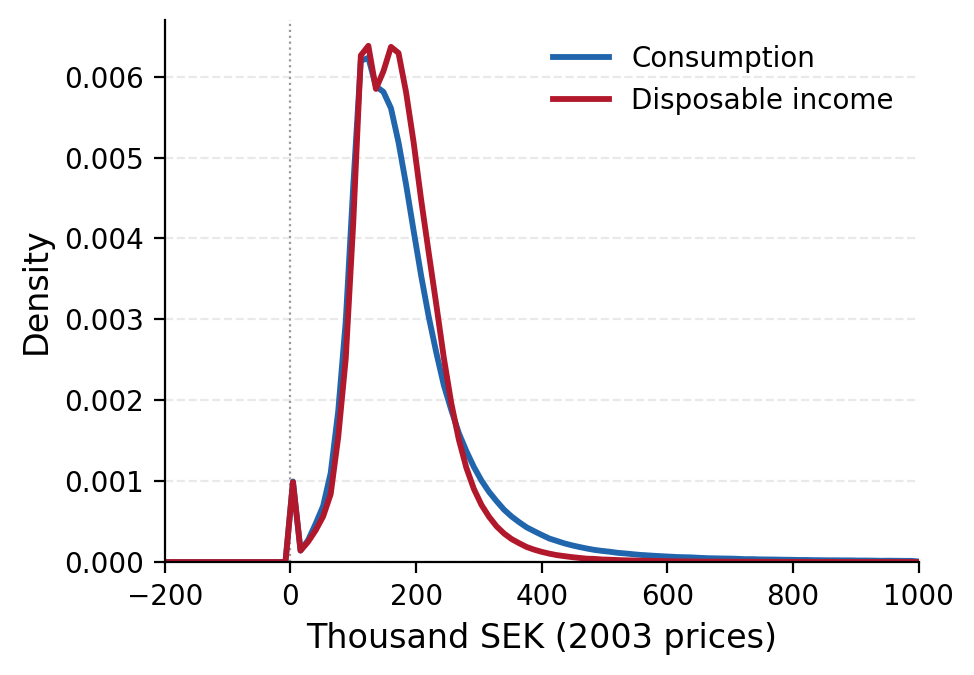}
        \caption{Consumption and income, singles}
    \end{subfigure}

    \par\medskip

    \begin{subfigure}[t]{0.44\textwidth}
        \centering
        \includegraphics[width=\linewidth]{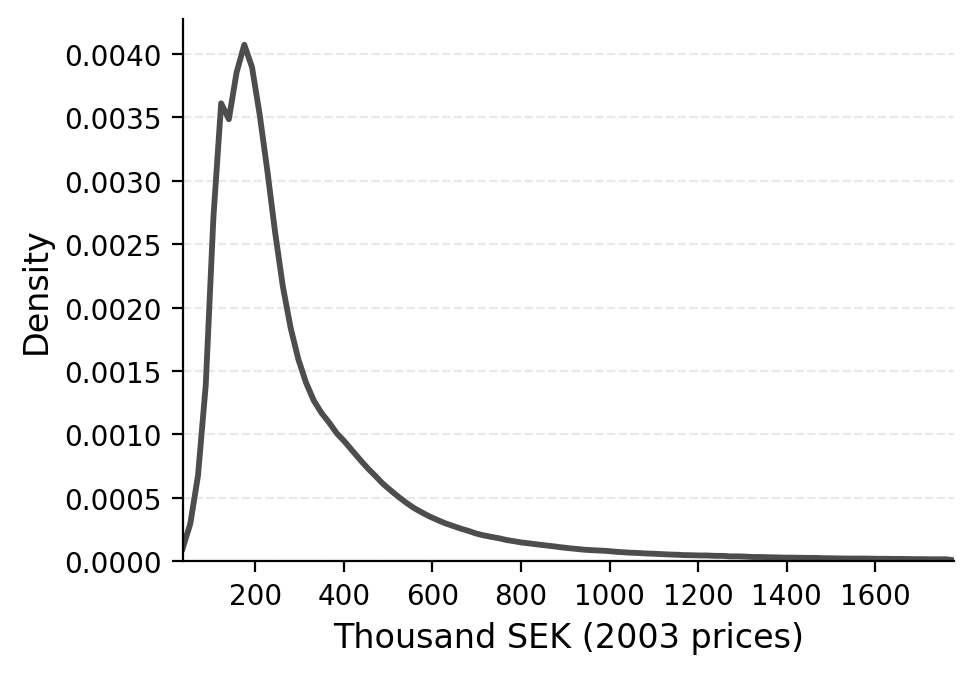}
        \caption{Lagged cash-on-hand, households}
    \end{subfigure}\quad
    \begin{subfigure}[t]{0.44\textwidth}
        \centering
        \includegraphics[width=\linewidth]{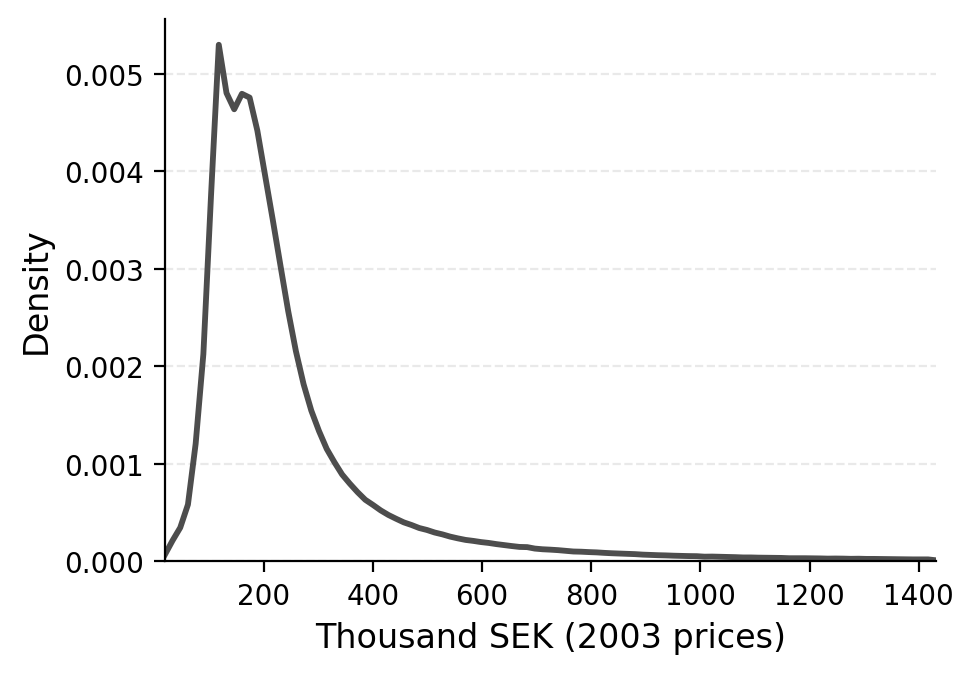}
        \caption{Lagged cash-on-hand, singles}
    \end{subfigure}

    \figurenote{Kernel density estimates are pooled over 2000--2007. Variables are in thousand 2003 SEK.}
\end{figure}

\section{Empirical Results}
\label{sec:empirical_results}

We first present results concerning the income process of the household and then estimates of the consumption pass-through coefficients and state-dependent MPCs. For all results using estimated pass-through coefficients, standard errors are bootstrapped, clustered at the level of the household. The bootstrap procedure is applied to the entire estimation procedure, from the extraction of the shocks through the second-stage estimation of Equation \eqref{eq:second_stage_estimation}. The resulting standard errors therefore account for estimation error in the generated regressors, the Kalman-smoothed shocks, as well as uncertainty about the distribution of the second-stage OLS residuals. For the income-process estimates, we report conventional analytical standard errors.  

\subsection{Income-Process Estimates and the State Variable Distribution}

Table~\ref{tab:income_process_moments} reports the income-process moments and baseline MA(1) parameter estimates. The estimates are very similar in the Singles and Households samples. 

Income growth has a large contemporaneous variance and negative first and second autocovariances, with the second smaller in absolute value. We conclude that the lag length of the MA(k) transitory income process is therefore at most 2. We select the MA(1) specification as our baseline, and show that our results are robust to selecting an MA(2) specification in Appendix~\ref{app:additional_results}. Since transitory persistence determines the width of the MPC bounds in Proposition~\ref{prop:mpc_bounds}, the documented small persistence in transitory income means that we can retrieve tight MPC bounds from the pass-through coefficients. 

\begin{table}[t]
    \centering
    \caption{Income-process moments and baseline MA(1) parameter estimates}
    \label{tab:income_process_moments}
    \begin{minipage}[t]{0.48\textwidth}
        \centering
        \caption*{A. Singles: moments}
        \resizebox{\linewidth}{!}{\begin{tabular}{l c c c}
\hline\hline
 & $\Delta \ln y_t$ & $\Delta \ln y_{t+1}$ & $\Delta \ln y_{t+2}$ \\
\hline
$\Delta \ln y_t$ & \shortstack[c]{0.0301\\(0.0299, 0.0303)} & \shortstack[c]{-0.0074\\(-0.0075, -0.0073)} & \shortstack[c]{-0.0026\\(-0.0027, -0.0025)} \\
$\Delta c_t$ & \shortstack[c]{0.0218\\(0.0216, 0.0220)} & \shortstack[c]{-0.0037\\(-0.0039, -0.0036)} & \shortstack[c]{-0.0022\\(-0.0024, -0.0021)} \\
\hline
\(N\) & \multicolumn{3}{c}{   1,724,757} \\
\hline\hline
\end{tabular}
}
    \end{minipage}\hfill
    \begin{minipage}[t]{0.48\textwidth}
        \centering
        \caption*{B. Singles: parameters}
        \resizebox{\linewidth}{!}{\begin{tabular}{l c c c}
\hline\hline
 & $\theta$ & $\sigma^2_{\epsilon}$ & $\sigma^2_{\eta}$ \\
\hline
Estimate & \shortstack[c]{0.2191\\(0.2140, 0.2242)} & \shortstack[c]{0.0123\\(0.0122, 0.0125)} & \shortstack[c]{0.0097\\(0.0094, 0.0099)} \\
\(N\) & \multicolumn{3}{c}{   1,724,757} \\
\hline\hline
\end{tabular}
}
    \end{minipage}

    \vspace{1em}

    \begin{minipage}[t]{0.48\textwidth}
        \centering
        \caption*{C. Households: moments}
        \resizebox{\linewidth}{!}{\begin{tabular}{l c c c}
\hline\hline
 & $\Delta \ln y_t$ & $\Delta \ln y_{t+1}$ & $\Delta \ln y_{t+2}$ \\
\hline
$\Delta \ln y_t$ & \shortstack[c]{0.0293\\(0.0292, 0.0294)} & \shortstack[c]{-0.0072\\(-0.0073, -0.0071)} & \shortstack[c]{-0.0026\\(-0.0026, -0.0025)} \\
$\Delta c_t$ & \shortstack[c]{0.0222\\(0.0221, 0.0223)} & \shortstack[c]{-0.0039\\(-0.0040, -0.0038)} & \shortstack[c]{-0.0023\\(-0.0024, -0.0022)} \\
\hline
\(N\) & \multicolumn{3}{c}{   3,060,626} \\
\hline\hline
\end{tabular}
}
    \end{minipage}\hfill
    \begin{minipage}[t]{0.48\textwidth}
        \centering
        \caption*{D. Households: parameters}
        \resizebox{\linewidth}{!}{\begin{tabular}{l c c c}
\hline\hline
 & $\theta$ & $\sigma^2_{\epsilon}$ & $\sigma^2_{\eta}$ \\
\hline
Estimate & \shortstack[c]{0.2213\\(0.2176, 0.2249)} & \shortstack[c]{0.0120\\(0.0119, 0.0121)} & \shortstack[c]{0.0095\\(0.0093, 0.0096)} \\
\(N\) & \multicolumn{3}{c}{   3,060,626} \\
\hline\hline
\end{tabular}
}
    \end{minipage}
\end{table}

These estimates are broadly aligned with previous evidence on the time-series properties of disposable income, but the magnitudes of both the variances and autocovariances are smaller compared to US survey evidence, which is expected when there is less measurement error. In the PSID-based estimates reported by \cite{Commault2022}, which use the same data as \cite{Blundell_et_al_2008_AER}, the autocovariances of annual log income growth have the same shape as in Table~\ref{tab:income_process_moments}: a large contemporaneous variance and negative covariances at the first two leads. Accordingly, \cite{Commault2022} selects an MA(1) specification.\footnote{Statistical significance is not a good criterion for selecting the lag length in administrative data, since the large sample size and low measurement error mean that the covariance at almost all leads is statistically significant.} Our Swedish covariance estimates are lower in absolute value, with a contemporaneous variance around 0.03 rather than 0.066 and first- and second-lead covariances roughly one-third to one-half as large. The larger covariances at short horizons in the PSID likely reflect i.i.d. measurement error in income. They are similarly comparable to the other cross-country survey-based disposable-income estimates in \cite{Krueger_et_al_2010_RED}, but with variances at the low end of the range.  

With an estimated income process, we apply the Kalman smoother and retrieve the underlying shocks. The level of permanent income is initialized using lagged observed disposable income and then constructed recursively from the recovered permanent shocks, \(\widehat P_{i,t}=\widehat P_{i,t-1}\exp(\widehat\eta_{i,t})\). With these estimates, we can then construct the central state variable used to estimate the heterogeneity in consumption response: cash-on-hand normalized by permanent income. Figure \ref{fig:normalized_coh_distr} shows the distribution of this variable. As expected, the distribution peaks between 1 and 2, indicating that cash-on-hand for the bulk of households is between one and two years of permanent income, and contains a fat tail with some households holding liquid wealth worth several years of permanent income.

\begin{figure}[t]
    \centering
    \caption{Distribution of lagged normalized cash-on-hand}
    \label{fig:normalized_coh_distr}
    \begin{subfigure}[t]{0.44\textwidth}
        \centering
        \includegraphics[width=\linewidth]{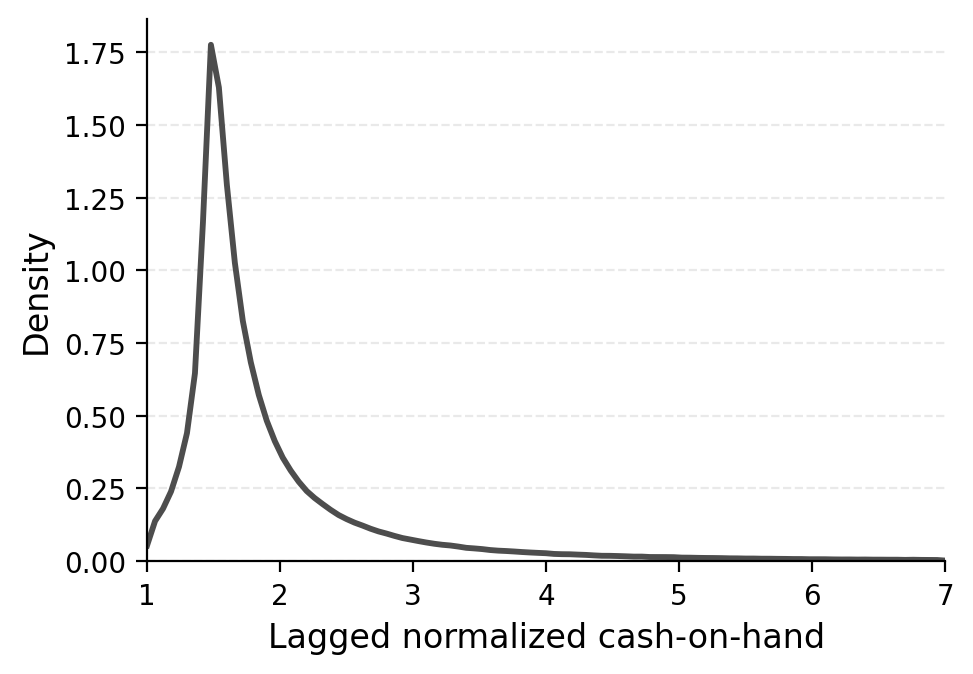}
        \caption{Households}
    \end{subfigure}\quad
    \begin{subfigure}[t]{0.44\textwidth}
        \centering
        \includegraphics[width=\linewidth]{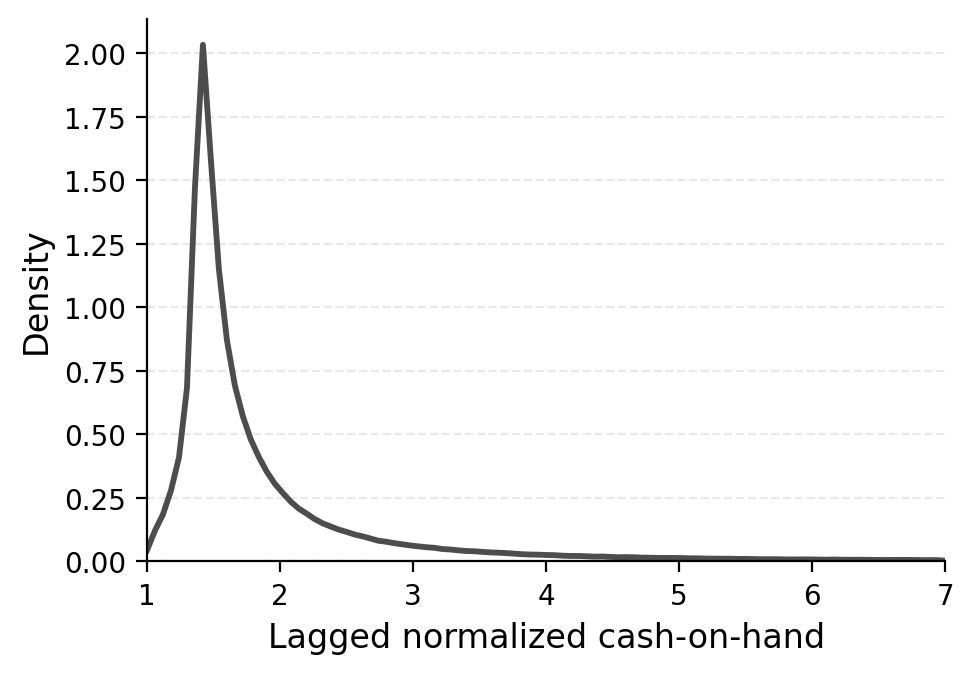}
        \caption{Singles}
    \end{subfigure}
    \figurenote{Kernel density estimates pool household-year observations over 2000--2007. Cash-on-hand is normalized by permanent income, constructed using the MA(1) income-process estimates. The horizontal axis can therefore be interpreted as multiples of yearly permanent income.}
\end{figure}

\subsection{Pass-Through Coefficients}

Under the MA(1) baseline, the predetermined state is $(m_{i,t-1},\varepsilon_{i,t-1})$. We partition observations into a $10\times 10$ grid using lagged normalized cash-on-hand and the lagged transitory shock. We retrieve the Kalman-smoothed permanent and transitory shock signals from the pooled sample, and then estimate Equation~\eqref{eq:second_stage_estimation} within these cells. The figures below collapse the grid to cash-on-hand profiles by averaging over the lagged-transitory-shock bins. The horizontal coordinate for each point is not the decile number, but the observation-count-weighted average of lagged normalized cash-on-hand, $m_{i,t-1}$, among observations in that cash-on-hand bin.

Figure~\ref{fig:pass_through_profiles_ma1} reports the baseline MA(1) pass-through profiles. The transitory pass-through coefficient declines with cash-on-hand in both samples, which is the direct source of the MPC gradient below. The permanent-shock pass-through is much flatter and is centered around one over most of the distribution. This pattern is consistent with the buffer-stock model interpretation: cash-on-hand matters most for the response to transitory resources, while permanent shocks shift lifetime resources and therefore have a more uniform effect across the distribution.

\begin{figure}[t]
    \centering
    \caption{Pass-through coefficients by cash-on-hand}
    \label{fig:pass_through_profiles_ma1}
    \includegraphics[width=0.92\textwidth]{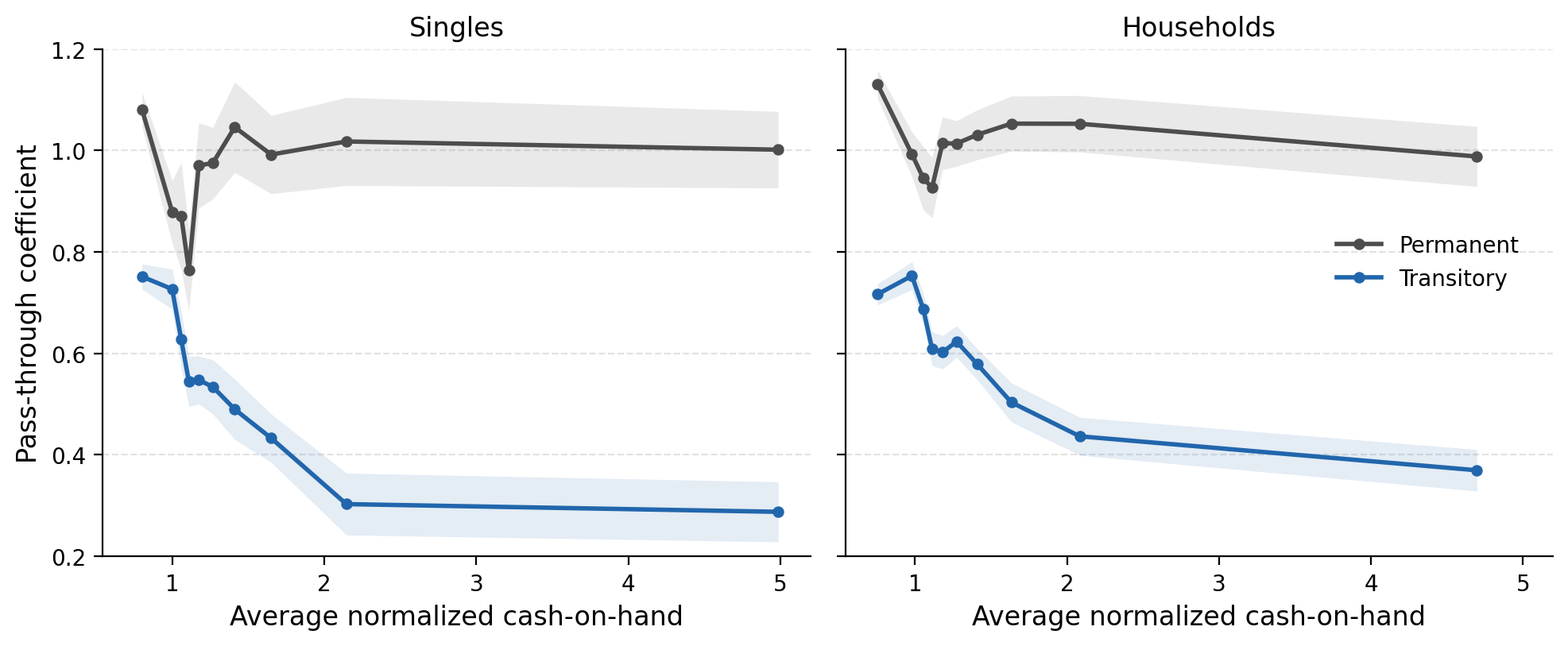}
    \figurenote{Panels report Singles and Households. Estimates are computed within decile bins of lagged normalized cash-on-hand; the horizontal axis plots the observation-count-weighted average normalized cash-on-hand in each bin. Each panel shows permanent- and transitory-shock pass-through coefficients from the OLS projection estimator. Shaded regions show 95 percent confidence intervals based on bootstrap standard errors clustered at the household level.}
\end{figure}

\subsection{MPC Gradients by Cash-on-Hand} \label{MPC_gradients}

We convert the transitory pass-through coefficient into the bounds on the MPC in Proposition~\ref{prop:mpc_bounds} and average across the lagged-shock bins within each cash-on-hand decile. The figures below report the lower bounds. Because $\theta_1$ is around 0.22 in both samples, the upper bounds are obtained by multiplying the plotted profiles by about 1.22. The slope of the profile is therefore not affected by the bound conversion: the conversion scales all deciles within a sample by the same factor.

Figure~\ref{fig:mpc_gradient_ma1} shows the main result. The projection estimator produces a steep downward yearly MPC profile for Singles. The profile is convex in average normalized cash-on-hand: most of the decline occurs at low cash-on-hand levels, while the profile is flatter over the upper part of the distribution. The lower (upper) bound is 0.68 (0.83) in the lowest cash-on-hand decile and about 0.31 (0.38) in the highest decile. The gradient is clearly downward sloping and precisely estimated. The population-average lower bound is about 0.47, with an upper bound of about 0.58. 

\begin{figure}[t]
    \centering
    \caption{MPC by cash-on-hand}
    \label{fig:mpc_gradient_ma1}
    \includegraphics[width=\textwidth]{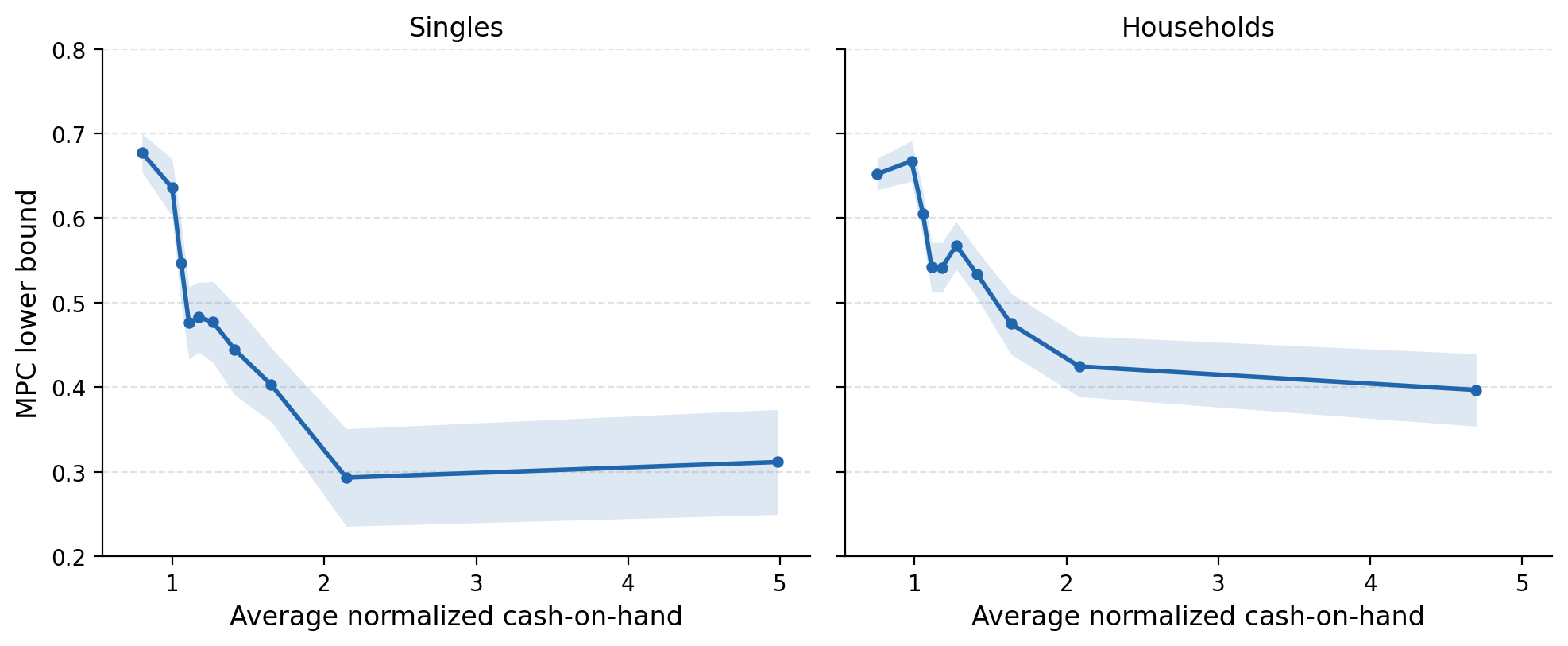}
    \figurenote{Estimates show lower bounds of the MPC, and are computed within decile bins of lagged normalized cash-on-hand; the horizontal axis plots the observation-count-weighted average normalized cash-on-hand in each bin. Shaded regions show 95 percent confidence intervals based on bootstrap standard errors clustered at the household level.}
\end{figure}

The Households panel repeats the exercise for the broader household sample. The level of the MPC is slightly higher in the upper part of the cash-on-hand distribution than in the Singles sample, but the main pattern is the same: a large drop at low levels of normalized cash-on-hand followed by a flatter profile further out in the distribution. The lower (upper) bound declines from about 0.65 (0.80) in the first decile to about 0.40 (0.48) in the tenth decile. 

The empirical gradient is the pattern predicted by buffer-stock consumption models. Households with little (permanent-income normalized) cash-on-hand have limited self-insurance and respond strongly to transitory income changes. Households with more cash-on-hand can smooth the same shock over time, so their current consumption response is smaller. The estimates therefore connect the pass-through equation in Proposition~\ref{prop:passthrough} to the key macroeconomic object emphasized by heterogeneous-agent models: the cross-sectional distribution of liquid resources shapes aggregate demand transmission.

At the same time, the level of the response, even though it refers to \emph{yearly} MPC, is substantial even at the upper end of the cash-on-hand distribution. This tendency for the top cash-on-hand decile is stronger for the Household sample (0.40) than the Singles sample (0.31). It is higher than one might expect from a single-asset, single-good model, where households with substantial liquid resources behave more like permanent-income consumers. However, our outcome is an expenditure measure rather than nondurable consumption alone, and it includes durable purchases such as cars and furniture. As emphasized by \cite{Maxted_Laibson_Moll_2025}, durable expenditures make the marginal propensity to spend significantly larger than the MPC predicted by a model with only one nondurable consumption good. They derive a conversion factor using a standard two-good model. Using their parameter values, the implied conversion factor between a yearly nondurable+durable MPX and the yearly nondurable MPC is about 1.5. Our baseline singles estimates would thus imply a gradient of the lower bound of the yearly nondurable MPC from about 0.45 in the first decile to 0.21 in the last decile, with a population average of 0.32, which is not far from what is implied by a standard calibration of a two-asset buffer-stock model \citep{Kaplan_Violante2014}.

\subsection{MPC Gradients: Income or Liquid Wealth}

Permanent-income normalized cash-on-hand is the sum of normalized liquid wealth and normalized disposable income. Figure~\ref{fig:mpc_component_gradients_ma1} partitions observations by the permanent-income-normalized components of lagged cash-on-hand while continuing to average over the remaining state variables. The liquid-wealth gradient is steep in both samples. In the Singles sample, the lower-bound MPC falls from about 0.70 at the bottom of the liquid-wealth-component distribution to close to zero at the top. In the Households sample, the same gradient declines from about 0.77 to about 0.24. That fact the number of points is fewer than 10 reflects that liquid wealth is below the detection limit of 10 000 SEK for a large share of the sample, consistent with the descriptive statistics in Table~\ref{tab:summary_stats}.

There is less variation in MPC across the disposable-income component, and the profile is non-monotone. For both Singles and Households, the lower-bound MPC falls from the bottom of the normalized disposable-income-component distribution through the middle deciles and then rises again at the top. This comparison suggests that the main driver of the cash-on-hand gradient in Figure~\ref{fig:mpc_gradient_ma1} is the gradient in liquid wealth, rather than the gradient in disposable income.

\begin{figure}[t]
    \centering
    \caption{MPC by components of cash-on-hand}
    \label{fig:mpc_component_gradients_ma1}
    \includegraphics[width=0.92\textwidth]{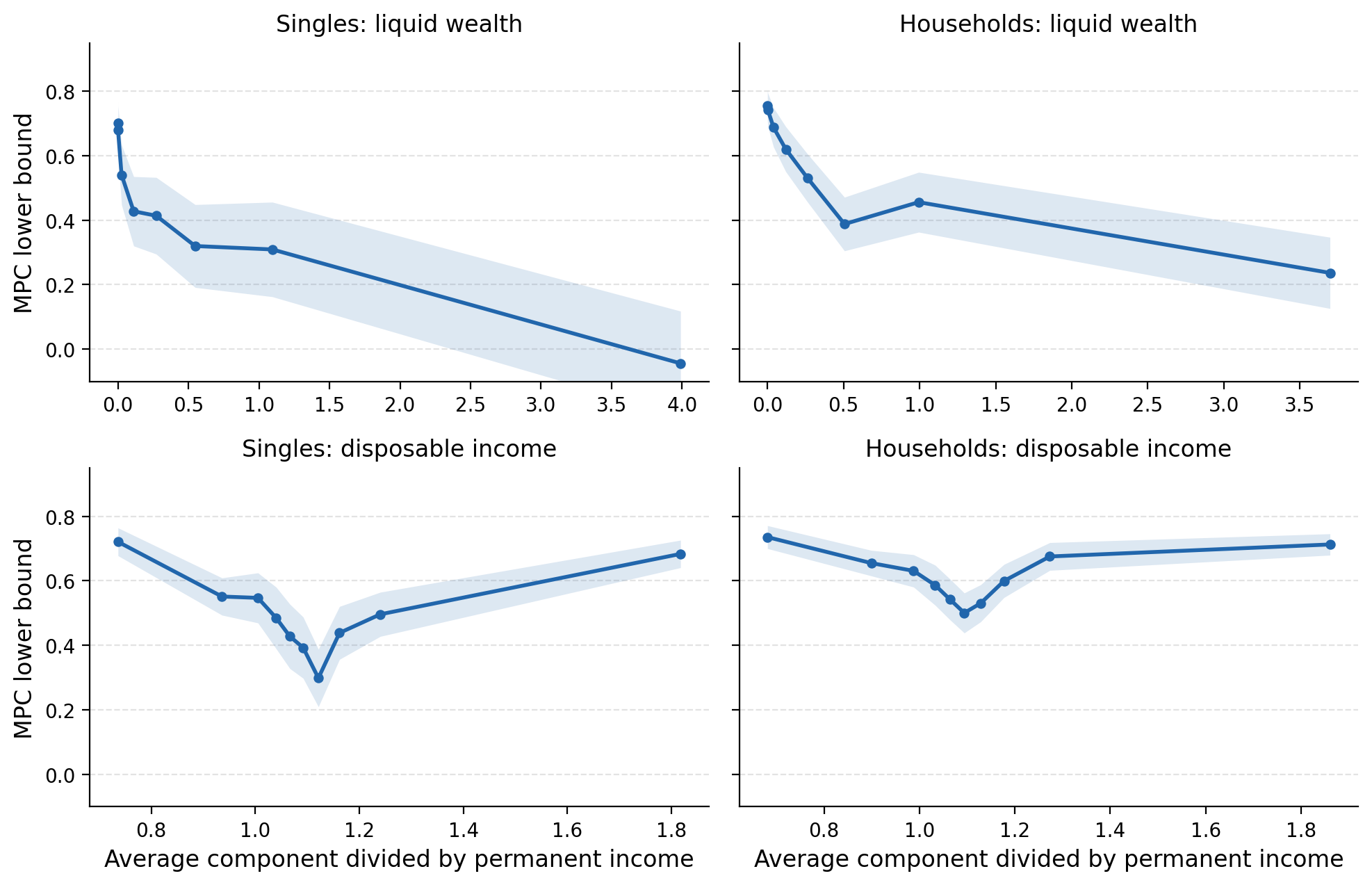}
    \figurenote{Estimates show lower bounds of the MPC, and are computed within decile bins of the relevant permanent-income-normalized cash-on-hand component; the horizontal axis plots the observation-count-weighted average component divided by permanent income in each bin. Shaded regions show 95 percent confidence intervals based on bootstrap standard errors clustered at the household level.}
\end{figure}

\subsection{Comparison to BPP-C}
\label{subsec:bpp_comparison}

Figure~\ref{fig:mpc_gradient_bpp_comparison_singles} compares the baseline projection estimator for Singles with the corresponding BPP-C estimator, which uses $\Delta \log Y_{i,t+2}$ as the single instrument for estimating the pass-through coefficient under the MA(1) baseline. We only show it for the Singles sample as the bootstrapping is computationally very time-consuming for the estimator, the corresponding graph for households without bootstrapped standard errors is shown in Appendix~\ref{app:additional_results}. This estimator is designed to be robust to income measurement error, but it uses a weaker income signal. Consistent with Proposition~\ref{prop:estimation_minvar}, the IV profile is much less precise---the standard errors span the entire range of the y-axis--and we cannot reject a flat profile. The large standard errors reflect the weak identifying covariance underlying the BPP-C estimator. With an MA(1) transitory income process, the distant-lead covariance used by BPP-C is proportional to the persistence parameter $\theta_1$, which is only about 0.22 in our sample.\footnote{As $\theta_1$ approaches zero, both the numerator and denominator of the covariance-ratio estimator are dominated by sampling variation; in the zero-persistence limit, the estimator behaves like a ratio of zero-mean normal variables and has Cauchy-like heavy tails with no finite mean.}

\begin{figure}[t]
    \centering
    \caption{MPC by cash-on-hand: OLS-LP versus BPP}
    \label{fig:mpc_gradient_bpp_comparison_singles}
    \includegraphics[width=0.78\textwidth]{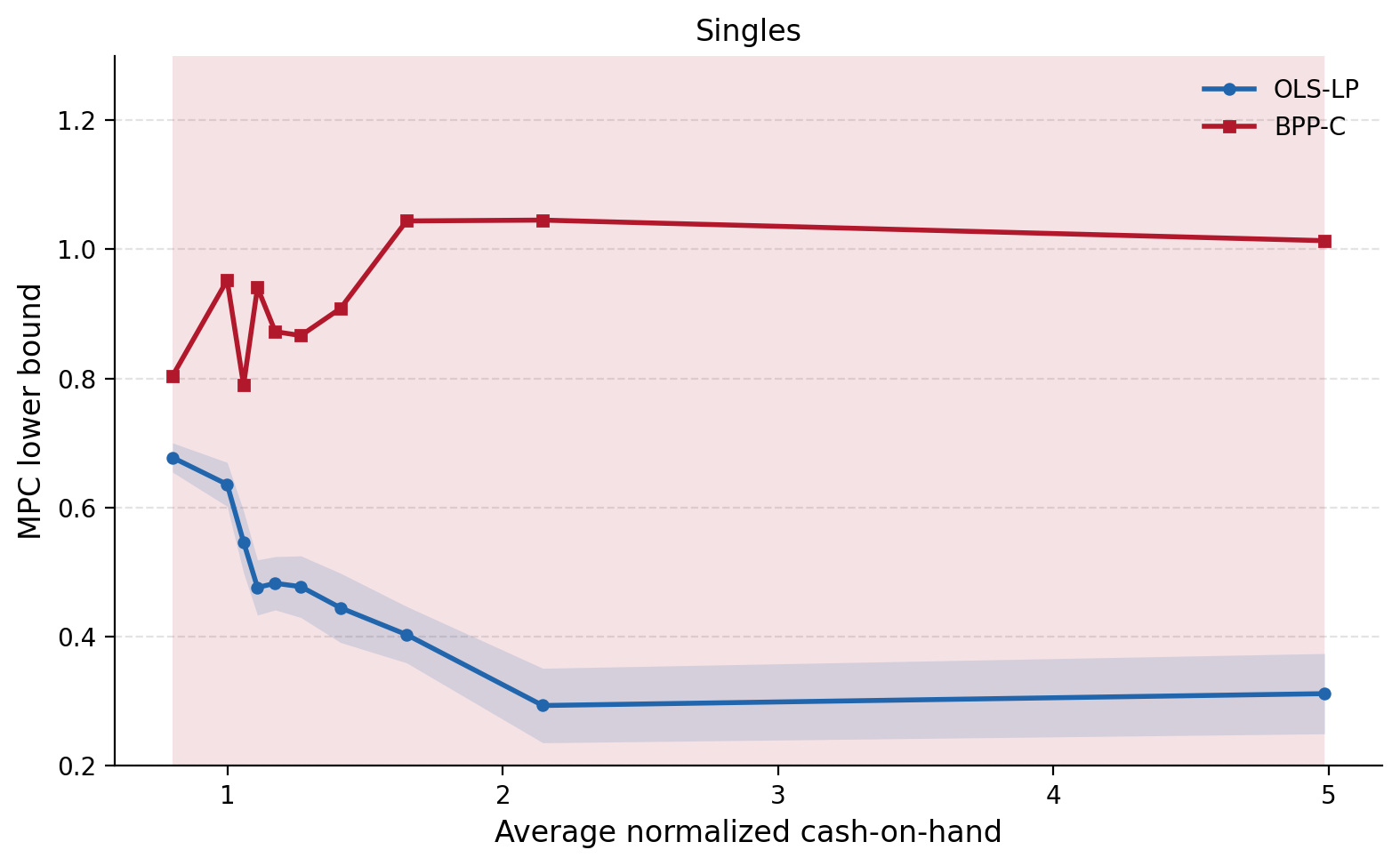}
    \figurenote{Estimates show lower bounds of the MPC for Singles, and are computed within decile bins of lagged normalized cash-on-hand; the horizontal axis plots the observation-count-weighted average normalized cash-on-hand in each bin. OLS-LP uses the Kalman-smoothed transitory shock. BPP-C uses the distant future-income-growth instrument. Shaded regions show 95 percent confidence intervals based on bootstrap standard errors clustered at the household level.}
\end{figure}

\subsection{Comparison to previous estimates}

Our baseline (Singles sample) average estimated MPC (nondurable MPC) lower and upper bounds are about 0.47 (0.32) and 0.58 (0.39); in the Households sample, the corresponding bounds are about 0.54 (0.36) and 0.66 (0.44). The average nondurable MPC estimates are above older estimates on PSID data using semi-structural methods (e.g. \citet{Blundell_et_al_2008_AER}), but in line with \cite{Commault2022}, who reports an MPC of 0.32. The total expenditure MPC estimates are also in line with semi-structural estimates from administrative data, such as \cite{Crawley2023}, which reports an annual MPC of 0.64. They are on the lower end compared to the literature on tax rebates; \cite{Johnson_et_al_2006_AER} estimate that households spent about two-thirds of the tax rebate within two quarters, while \cite{Parker_2013_AER} find even larger total expenditure responses to the 2008 stimulus payments once durable purchases are included.

The distinctive feature of our estimates is the steepness of the cash-on-hand gradient. Qualitatively, this is consistent with the broader evidence that liquid resources are central for consumption smoothing. However, the gradient here is steeper compared to many papers, including both older survey-based estimates \citep{Johnson_et_al_2006_AER,Parker_2013_AER} and more recent papers using special income shocks and experimental designs \citep{Fuster_Kaplan_Zafar2020,Fagareng_et_al_2021_AEJM,Boehm2025}. Our estimates are more similar to \cite{Crawley2023}, who find a strong negative relationship between liquid wealth and the total expenditure MPC in Danish administrative data, and \cite{Ganong2025}, who find that consumption smoothing of typical labor-income shocks improves sharply with liquid wealth in US transaction data. In numbers: the lower-bound MPC estimates in our data fall by about 0.37 from the first to the last cash-on-hand decile in the baseline sample, which is very close to the fall in \cite{Ganong2025} across liquid-wealth deciles, and somewhat smaller than the 0.48 decline across liquid-wealth quintiles in \cite{Crawley2023}. 

The pattern with a flatter gradient in studies using special income shocks and experimental designs suggests that the source of income variation may matter for how households respond. Ordinary earnings fluctuations are part of the income risk that households routinely face, so households with substantial liquid resources may be able to smooth them more effectively. Special payments (lottery wins, tax rebates), by contrast, are typically positive, salient, sometimes large, and often tied to unusual timing or design features; they may also trigger durable purchases to a larger extent and be tied to specific behavioral spending channels (especially for lottery winnings). These features can raise spending responses among liquid or affluent households and thereby flatten the estimated gradient by liquid resources, as suggested by the role of new vehicles in \cite{Misra2014} and the transfer-design effects in \cite{Boehm2025}. This interpretation is necessarily cautious, since the studies also differ in their horizons, consumption concepts, shock sizes, and wealth definitions.

\section{Conclusion}

This paper develops a semi-structural method for estimating state-dependent MPCs in administrative data with precisely measured income histories. We derive a pass-through equation from a canonical buffer-stock model, show how transitory-shock pass-through bounds the MPC conditional on household state variables, and estimate the relevant income shocks using a Kalman smoother. Under negligible income measurement error, this estimator is identified and has minimum variance within the class of consistent estimators that are linear in household income histories, including BPP-style estimators.

Applying the method to Swedish administrative data, we find a steep negative MPC gradient with respect to cash-on-hand normalized by permanent income. The average yearly MPC is high, around one half, but the response is much larger among households with little cash-on-hand than among households at the top of the distribution: in the baseline Singles sample, the lower bound falls from about 0.68 in the bottom decile to about 0.31 in the top decile. The pattern is similar for the broader household sample and robust to richer income-process specifications. This suggests that liquidity is an important determinant of consumption responses to ordinary income risk.

\putbib[References/ref]
\end{bibunit}

\newpage

\appendix
\begin{bibunit}

\section{Proofs for Identification and MPC Bounds}
\label{app:identification_proofs}

\subsection{Pass-Through Equation}
\begin{proof}[Proof of Proposition \ref{prop:passthrough}]
Fix the predetermined observed state $(m_{i,t-1},\mathbf{s}_{i,t-1})$. From \eqref{eq:consumption_growth_policy}, consumption growth can be written as
\[
    \Delta \log C_{i,t}
    =
    \eta_{i,t}
    +\log c(m_{i,t},\mathbf{s}_{i,t})
    -\log c(m_{i,t-1},\mathbf{s}_{i,t-1}),
\]
where $m_{i,t}$ is given by \eqref{eq:normalized_cash_on_hand_lom} and $\mathbf{s}_{i,t}$ is the MA state implied by $\mathbf{s}_{i,t-1}$ and the new transitory innovation $\varepsilon_{i,t}$.

Linearize this expression with respect to $(\eta_{i,t},\varepsilon_{i,t})$ around $\eta_{i,t}=\varepsilon_{i,t}=0$. Let $\nu^0_{i,t}\equiv \sum_{j=1}^k\theta_j\varepsilon_{i,t-j}$ denote the predetermined transitory component, and let $\bar{\mathbf{s}}(\mathbf{s}_{i,t-1})$ denote the period-$t$ MA state when $\varepsilon_{i,t}=0$. The derivatives of normalized cash-on-hand at this point are
\[
    \left.
    \frac{\partial m_{i,t}}{\partial \eta_{i,t}}
    \right|_0
    =
    -(1+r)
    \left(m_{i,t-1}-c(m_{i,t-1},\mathbf{s}_{i,t-1})\right),
    \qquad
    \left.
    \frac{\partial m_{i,t}}{\partial \varepsilon_{i,t}}
    \right|_0
    =
    Y^d_{i,t}\exp(\nu^0_{i,t}),
\]
where the second equality uses $\theta_0=1$.

The coefficient on the permanent shock is therefore
\[
    \lambda(m_{i,t-1},\mathbf{s}_{i,t-1})
    =
    1
    -(1+r)
    \left(m_{i,t-1}-c(m_{i,t-1},\mathbf{s}_{i,t-1})\right)
    \frac{\partial \log c(m_{i,t},\mathbf{s}_{i,t})}{\partial m_{i,t}},
\]
with the derivative evaluated at
\[
    (m_{i,t},\mathbf{s}_{i,t})
    =
    \left(
    g(m_{i,t-1},\mathbf{s}_{i,t-1},0,\nu^0_{i,t}),
    \bar{\mathbf{s}}(\mathbf{s}_{i,t-1})
    \right).
\]
The coefficient on the transitory shock is
\[
    \gamma(m_{i,t-1},\mathbf{s}_{i,t-1})
    =
    Y^d_{i,t}\exp(\nu^0_{i,t})
    \frac{\partial \log c(m_{i,t},\mathbf{s}_{i,t})}{\partial m_{i,t}}
    +
    \frac{\partial \log c(m_{i,t},\mathbf{s}_{i,t})}{\partial s^{(1)}_{i,t}},
\]
evaluated at the same point. This gives \eqref{eq:passthrough_equation} and \eqref{eq:gamma_cash_news}.
\end{proof}

\subsection{MPC Bound from Pass-Through}
\begin{proof}[Proof of Proposition \ref{prop:mpc_bounds}]
Let
\[
    \Xi_{i,t}
    \equiv
    Y^d_{i,t}\exp(\nu^0_{i,t})
    \frac{\partial \log c(m_{i,t},\mathbf{s}_{i,t})}{\partial m_{i,t}},
\]
with all objects evaluated at the linearization point in Proposition \ref{prop:passthrough}. This is the log-consumption response generated by the cash-on-hand channel of a transitory innovation. By \eqref{eq:gamma_cash_news},
\begin{equation}
    \gamma(m_{i,t-1},\mathbf{s}_{i,t-1})
    =
    \Xi_{i,t}
    +
    \frac{\partial \log c(m_{i,t},\mathbf{s}_{i,t})}{\partial s^{(1)}_{i,t}}.
\label{eq:appendix_gamma_xi_decomposition}
\end{equation}
Assumption \ref{ass:bounded_news} implies
\begin{equation}
    \Xi_{i,t}
    \le
    \gamma(m_{i,t-1},\mathbf{s}_{i,t-1})
    \le
    \left(1+\sum_{j=1}^{k}\theta_j\right)\Xi_{i,t}.
\label{eq:appendix_gamma_xi_bounds}
\end{equation}

It remains to translate $\Xi_{i,t}$ into the MPC. At the linearization point,
\[
    C_{i,t}=P_{i,t}c(m_{i,t},\mathbf{s}_{i,t}),
    \qquad
    Y_{i,t}=Y^d_{i,t}P_{i,t}\exp(\nu^0_{i,t}).
\]
Hence
\begin{equation}
    \frac{C_{i,t}}{Y_{i,t}}\Xi_{i,t}
    =
    \frac{c(m_{i,t},\mathbf{s}_{i,t})}
         {Y^d_{i,t}\exp(\nu^0_{i,t})}
    Y^d_{i,t}\exp(\nu^0_{i,t})
    \frac{1}{c(m_{i,t},\mathbf{s}_{i,t})}
    \frac{\partial c(m_{i,t},\mathbf{s}_{i,t})}{\partial m_{i,t}}
    =
    \frac{\partial c(m_{i,t},\mathbf{s}_{i,t})}{\partial m_{i,t}}.
\label{eq:appendix_xi_mpc_identity}
\end{equation}
Since $m_{i,t}=M_{i,t}/P_{i,t}$ and $C_{i,t}=P_{i,t}c(m_{i,t},\mathbf{s}_{i,t})$, the last expression in \eqref{eq:appendix_xi_mpc_identity} is the MPC with respect to current cash-on-hand. Multiplying \eqref{eq:appendix_gamma_xi_bounds} by $C_{i,t}/Y_{i,t}$ and substituting \eqref{eq:appendix_xi_mpc_identity} gives \eqref{eq:mpc_bounds}.
\end{proof}

\section{State-Space and Kalman Smoother Details}
\label{app:kalman_smoother}

\subsection{Setup}

We illustrate the Kalman smoother using the baseline MA(1) case. The income process implies
\[
    \Delta \log Y_{i,t}
    =
    \eta_{i,t}
    +
    \psi_0\varepsilon_{i,t}
    +
    \psi_1\varepsilon_{i,t-1}
    +
    \psi_2\varepsilon_{i,t-2},
\]
where
\[
    \psi_0=1,\qquad
    \psi_1=\theta_1-1,\qquad
    \psi_2=-\theta_1.
\]
The same logic applies for general $k$, with income growth having an MA$(k+1)$ transitory component.

Define
\[
    \boldsymbol{\alpha}_{i,t}
    =
    \begin{pmatrix}
    \eta_{i,t}\\
    \varepsilon_{i,t}\\
    \varepsilon_{i,t-1}\\
    \varepsilon_{i,t-2}
    \end{pmatrix}.
\]
Then the MA(1) model can be written as
\[
    \boldsymbol{\alpha}_{i,t}=T\boldsymbol{\alpha}_{i,t-1}+R\mathbf{u}_{i,t},
    \qquad
    \mathbf{u}_{i,t}
    =
    \begin{pmatrix}
    \eta_{i,t}\\
    \varepsilon_{i,t}
    \end{pmatrix},
    \qquad
    \operatorname{Var}(\mathbf{u}_{i,t})
    =
    \begin{pmatrix}
    \sigma_\eta^2&0\\
    0&\sigma_\varepsilon^2
    \end{pmatrix},
\]
and
\[
    T
    =
    \begin{pmatrix}
    0&0&0&0\\
    0&0&0&0\\
    0&1&0&0\\
    0&0&1&0
    \end{pmatrix},
    \qquad
    R
    =
    \begin{pmatrix}
    1&0\\
    0&1\\
    0&0\\
    0&0
    \end{pmatrix}.
\]
The observation equation is
\[
    \Delta\log Y_{i,t}=\mathbf{Z}\boldsymbol{\alpha}_{i,t},
    \qquad
    \mathbf{Z}=(1,\psi_0,\psi_1,\psi_2).
\]
Here $\mathbf{Z}$ is the $1\times4$ row vector that maps the state into observed income growth. The matrices $T$ and $R$ shift lagged transitory shocks through the state vector and inject the new permanent and transitory innovations.

\subsection{Kalman Filter and Smoother Implementation}
\label{app:kalman_implementation}
Given $(\theta_1,\sigma_\varepsilon^2,\sigma_\eta^2)$, the objects computed by the filter are the one-sided linear prediction of the state and its prediction-error covariance. Let
\[
    \widehat{\boldsymbol{\alpha}}_{i,t\mid s}
    \equiv
    \operatorname{Proj}(\boldsymbol{\alpha}_{i,t}\mid \Delta\log Y_{i,1:s}),
    \qquad
    P_{t\mid s}
    \equiv
    \operatorname{Var}(\boldsymbol{\alpha}_{i,t}-\widehat{\boldsymbol{\alpha}}_{i,t\mid s}).
\]
Under joint normality, $\widehat{\boldsymbol{\alpha}}_{i,t\mid s}$ is also the conditional expectation $E(\boldsymbol{\alpha}_{i,t}\mid \Delta\log Y_{i,1:s})$.
The one-step-ahead prediction from information through $t-1$ is
\begin{align}
    \widehat{\boldsymbol{\alpha}}_{i,t\mid t-1}
    &=
    T\widehat{\boldsymbol{\alpha}}_{i,t-1\mid t-1},
    &
    P_{t\mid t-1}
    &=
    TP_{t-1\mid t-1}T'
    +
    R\operatorname{Var}(\mathbf{u}_{i,t})R'.
    \label{eq:kalman_prediction}
\end{align}
After observing $\Delta\log Y_{i,t}$, the filter computes the forecast error, its variance, and the Kalman gain,
\[
    v_{i,t}
    =
    \Delta\log Y_{i,t}
    -
    \mathbf{Z}\widehat{\boldsymbol{\alpha}}_{i,t\mid t-1},
    \qquad
    F_t
    =
    \mathbf{Z}P_{t\mid t-1}\mathbf{Z}',
    \qquad
    \mathbf{K}_t
    =
    P_{t\mid t-1}\mathbf{Z}'F_t^{-1}.
\]
Here $\mathbf{K}_t$ is the $4\times1$ Kalman-gain vector that maps the scalar forecast error into the state update.
The filtered state is then
\begin{align}
    \widehat{\boldsymbol{\alpha}}_{i,t\mid t}
    &=
    \widehat{\boldsymbol{\alpha}}_{i,t\mid t-1}
    +
    \mathbf{K}_tv_{i,t},
    &
    P_{t\mid t}
    &=
    P_{t\mid t-1}
    -
    \mathbf{K}_tF_t\mathbf{K}_t'.
    \label{eq:kalman_update}
\end{align}
Thus the filter uses only current and past income growth. It delivers the best one-sided estimate of the state at each date, together with the uncertainty in that estimate.

The fixed-interval smoother then works backward from $T_i$ to date 1. Define
\[
    J_t
    =
    P_{t\mid t}T'P_{t+1\mid t}^{-1}.
\]
The Rauch--Tung--Striebel recursion computes the two-sided state mean and covariance as
\begin{align}
    \widehat{\boldsymbol{\alpha}}_{i,t\mid T_i}
    &=
    \widehat{\boldsymbol{\alpha}}_{i,t\mid t}
    +
    J_t
    \left(
    \widehat{\boldsymbol{\alpha}}_{i,t+1\mid T_i}
    -
    \widehat{\boldsymbol{\alpha}}_{i,t+1\mid t}
    \right),
    \label{eq:rts_smoother_mean}\\
    P_{t\mid T_i}
    &=
    P_{t\mid t}
    +
    J_t
    \left(
    P_{t+1\mid T_i}
    -
    P_{t+1\mid t}
    \right)
    J_t'.
    \label{eq:rts_smoother_variance}
\end{align}
This second pass is what makes future income growth informative about the date-$t$ shock: observations after $t$ help distinguish a persistent permanent innovation from the finite MA footprint of a transitory innovation.

The recovered shocks used in the second stage are the first two elements of the smoothed state,
\begin{equation}
    \widehat{\eta}_{i,t}
    =
    \mathbf{e}_1'\widehat{\boldsymbol{\alpha}}_{i,t\mid T_i},
    \qquad
    \widehat{\varepsilon}_{i,t}
    =
    \mathbf{e}_2'\widehat{\boldsymbol{\alpha}}_{i,t\mid T_i},
    \label{eq:kalman_recovered_shocks}
\end{equation}
where $\mathbf{e}_1$ and $\mathbf{e}_2$ select the permanent and current transitory innovations from $\boldsymbol{\alpha}_{i,t}$. These are the regressors in Equation \eqref{eq:second_stage_estimation}.

Under joint normality, these smoothed states are conditional expectations. Without normality, the same algorithm computes the best linear predictions implied by the second moments of the state-space model.

\section{Minimum-Variance Linear Signal Result}
\label{app:minvar_signal}

\subsection{Setup and Admissible Estimator Class}
This appendix states the projection argument behind Propositions \ref{prop:passthrough_identification} and \ref{prop:estimation_minvar}. Let
\[
    \mathbf{y}_i
    =
    (\Delta\log Y_{i,1},\ldots,\Delta\log Y_{i,T_i})'
\]
denote the observed income-growth history for household $i$. Fix a date $t$ and let $\varepsilon_{i,t}$ denote the latent transitory shock. The permanent-shock component and the state controls are either included directly in the regression or partialled out. To keep notation compact, write the remaining pass-through equation as
\begin{equation}
    \widetilde{\Delta\log C}_{i,t}
    =
    \gamma\varepsilon_{i,t}
    +
    e_{i,t}.
\label{eq:appendix_residual_passthrough}
\end{equation}
The projection signal is
\begin{equation}
    \widehat{\varepsilon}_{i,t}
    \equiv
    \operatorname{Proj}(\varepsilon_{i,t}\mid \mathbf{y}_i).
\label{eq:appendix_projection_signal}
\end{equation}
Writing $\Sigma_y\equiv\operatorname{Var}(\mathbf{y}_i)$ and $\mathbf{c}_t\equiv\operatorname{Cov}(\mathbf{y}_i,\varepsilon_{i,t})$, the projection in \eqref{eq:appendix_projection_signal} is
\begin{equation}
    \widehat{\varepsilon}_{i,t}
    =
    \mathbf{c}_t'\Sigma_y^{-1}\mathbf{y}_i.
\label{eq:appendix_projection_signal_formula}
\end{equation}

The comparison class consists of IV or GMM procedures whose identifying signal is linear in the income history. For a scalar signal $Z_{i,t}=\mathbf{a}'\mathbf{y}_i$, the population IV estimand is
\begin{equation}
    \gamma(Z)
    =
    \frac{\operatorname{Cov}(Z_{i,t},\widetilde{\Delta\log C}_{i,t})}
         {\operatorname{Cov}(Z_{i,t},\varepsilon_{i,t})}.
\label{eq:appendix_scalar_iv_estimand}
\end{equation}
This notation also covers linear GMM procedures. If $\mathbf{Z}_{i,t}$ is a vector of linear income-history moments and $W$ is the population weighting matrix, the corresponding one-parameter GMM estimand is the ratio in \eqref{eq:appendix_scalar_iv_estimand} evaluated at the composite scalar signal
\begin{equation}
    S_{i,t}
    =
    \operatorname{Cov}(\mathbf{Z}_{i,t},\varepsilon_{i,t})'W\mathbf{Z}_{i,t}.
\label{eq:appendix_gmm_composite_signal}
\end{equation}
Thus, within the admissible class, GMM differs from scalar IV only in the choice of linear combination of income-history signals.

We maintain that, after conditioning on the state and the permanent-shock component, $e_{i,t}$ is orthogonal to all such income-history signals.

\subsection{Identification of the Projection Estimand}
\begin{proof}[Proof of Proposition \ref{prop:passthrough_identification}]
After partialling out the permanent shock and state controls, the population coefficient on $\widehat{\varepsilon}_{i,t}$ is
\begin{equation}
    \frac{
    \operatorname{Cov}(\widehat{\varepsilon}_{i,t},
    \widetilde{\Delta\log C}_{i,t})
    }{
    \operatorname{Var}(\widehat{\varepsilon}_{i,t})
    }.
\label{eq:appendix_projection_coefficient}
\end{equation}
Using the maintained pass-through equation \eqref{eq:appendix_residual_passthrough}, the coefficient in \eqref{eq:appendix_projection_coefficient} equals
\begin{equation}
    \gamma
    \frac{
    \operatorname{Cov}(\widehat{\varepsilon}_{i,t},\varepsilon_{i,t})
    }{
    \operatorname{Var}(\widehat{\varepsilon}_{i,t})
    }
    +
    \frac{
    \operatorname{Cov}(\widehat{\varepsilon}_{i,t},e_{i,t})
    }{
    \operatorname{Var}(\widehat{\varepsilon}_{i,t})
    }.
\label{eq:appendix_projection_coefficient_decomposition}
\end{equation}
The second term in \eqref{eq:appendix_projection_coefficient_decomposition} is zero because $\widehat{\varepsilon}_{i,t}$ is a linear function of $\mathbf{y}_i$ and the residual is orthogonal to admissible income-history signals. The first ratio in \eqref{eq:appendix_projection_coefficient_decomposition} is one because the projection residual $\varepsilon_{i,t}-\widehat{\varepsilon}_{i,t}$ is orthogonal to every linear function of $\mathbf{y}_i$, including $\widehat{\varepsilon}_{i,t}$. Hence
\[
    \operatorname{Cov}(\widehat{\varepsilon}_{i,t},\varepsilon_{i,t})
    =
    \operatorname{Var}(\widehat{\varepsilon}_{i,t}).
\]
\end{proof}

If the permanent shock is omitted, the same calculation gives \eqref{eq:omitted_permanent_bias}, which is why the empirical specification includes the recovered permanent-shock component rather than treating the transitory signal as the only relevant income innovation.

\subsection{Minimum-Variance Linear Income-History Signal}
\begin{proof}[Proof of Proposition \ref{prop:estimation_minvar}]
For any admissible scalar instrument $Z_{i,t}=\mathbf{a}'\mathbf{y}_i$, the standard IV expansion gives
\begin{equation}
    \operatorname{Avar}(\widehat{\gamma}_0(Z))
    \propto
    \frac{\operatorname{Var}(Z_{i,t})}
         {\operatorname{Cov}(Z_{i,t},\varepsilon_{i,t})^2}
    =
    \frac{1}
         {\sigma_\varepsilon^2
         \rho^2(Z_{i,t},\varepsilon_{i,t})}.
\label{eq:appendix_scalar_instrument_avar}
\end{equation}
Thus minimizing asymptotic variance is equivalent to maximizing the squared correlation between $Z_{i,t}$ and $\varepsilon_{i,t}$.

For $Z_{i,t}=\mathbf{a}'\mathbf{y}_i$,
\[
    \rho^2(Z_{i,t},\varepsilon_{i,t})
    =
    \frac{(\mathbf{a}'\mathbf{c}_t)^2}
         {(\mathbf{a}'\Sigma_y \mathbf{a})\sigma_\varepsilon^2}.
\]
By Cauchy--Schwarz in the inner product induced by $\Sigma_y$,
\[
    (\mathbf{a}'\mathbf{c}_t)^2
    \le
    (\mathbf{a}'\Sigma_y \mathbf{a})(\mathbf{c}_t'\Sigma_y^{-1}\mathbf{c}_t),
\]
with equality if and only if $\mathbf{a}$ is proportional to $\Sigma_y^{-1}\mathbf{c}_t$. Therefore the maximal squared correlation is attained by
\begin{equation}
    Z^*_{i,t}
    =
    \mathbf{c}_t'\Sigma_y^{-1}\mathbf{y}_i
    =
    \widehat{\varepsilon}_{i,t}.
\label{eq:appendix_optimal_linear_signal}
\end{equation}
Substituting \eqref{eq:appendix_optimal_linear_signal} into \eqref{eq:appendix_scalar_instrument_avar} gives \eqref{eq:minvar_ratio}.
\end{proof}

\section{Measurement Error Simulations}
\label{app:measurement_error_simulations}

The simulations evaluate the premise that income is measured accurately enough for the projection estimator to be useful. For this exercise, we assume that transitory income is MA(1) in levels, as in our baseline sample, and take the income-process parameters from Table~\ref{tab:income_process_moments} and the estimated consumption-function parameters from Figure~\ref{fig:pass_through_profiles_ma1}. The robust BPP-C estimator therefore uses the two-period-ahead income growth instrument, $\Delta \log Y_{i,t+2}$. We compare BPP-C to a naive linear-projection estimator that ignores income measurement error when constructing the Kalman-smoothed shock. 

We model measurement error in log income as an i.i.d. classical shock $\xi_t$, so that observed income growth contains the MA(1) component $\xi_t-\xi_{t-1}$. The noise share is
\begin{equation}
    \frac{\operatorname{Var}(\xi_t-\xi_{t-1})}
    {\operatorname{Var}(\Delta \log y^{obs}_t)}
    =
    \frac{2\sigma_\xi^2}
    {\operatorname{Var}(\Delta \log y^{obs}_t)}.
\label{eq:income_noise_share}
\end{equation}
The scale in \eqref{eq:income_noise_share} is chosen to match the validation evidence. In survey settings such as the PSID, where income and consumption are collected separately, validation studies by \cite{Bound1994} and \cite{Pischke1995} imply substantial measurement error in earnings changes, roughly in the range of 15--20 percent. In Nordic administrative registers, the relevant range is much lower. \cite{KapteynYpma2007}, \cite{MeijerRohwedderWansbeek2012}, and \cite{Kristensen_Westergaard_Nielsen_2007} point to very accurate register earnings conditional on correct matching, suggesting that 0--2 percent is the empirically plausible range for population-wide register applications.

These validation facts motivate two simulation exercises. First, we present a simulation exercise that mimics survey data, where income and consumption measurement errors are i.i.d. and independent. The independent-error simulation varies the noise share from 0 to 20 percent, covering the survey range. Second, we present a simulation exercise that mimics budget-constraint imputation (as in our data), where income measurement error mechanically enters imputed consumption. The budget-constraint imputation simulation focuses on 0 to 5 percent, because the register concern is not mainly that income error is large. It is that any remaining income error enters imputed consumption mechanically through the household budget constraint, creating a cross-equation correlation between measurement error in observed income and observed consumption.

Figure \ref{fig:bias_variance_measurement_error} summarizes the simulation results. The figure reports the estimator behavior over repeated samples as the income-noise share varies. Each point in the figure is based on 1,000 simulated samples with 30,421 observations per sample. It illustrates the bias-variance trade-off: the projection estimator uses a stronger signal but is exposed to misspecification when income measurement error is ignored, while BPP-C remains robust but becomes imprecise.

\begin{figure}[t]
    \centering
    \caption{Bias-variance trade-off with income measurement error}
    \label{fig:bias_variance_measurement_error}
    \includegraphics[width=0.86\textwidth]{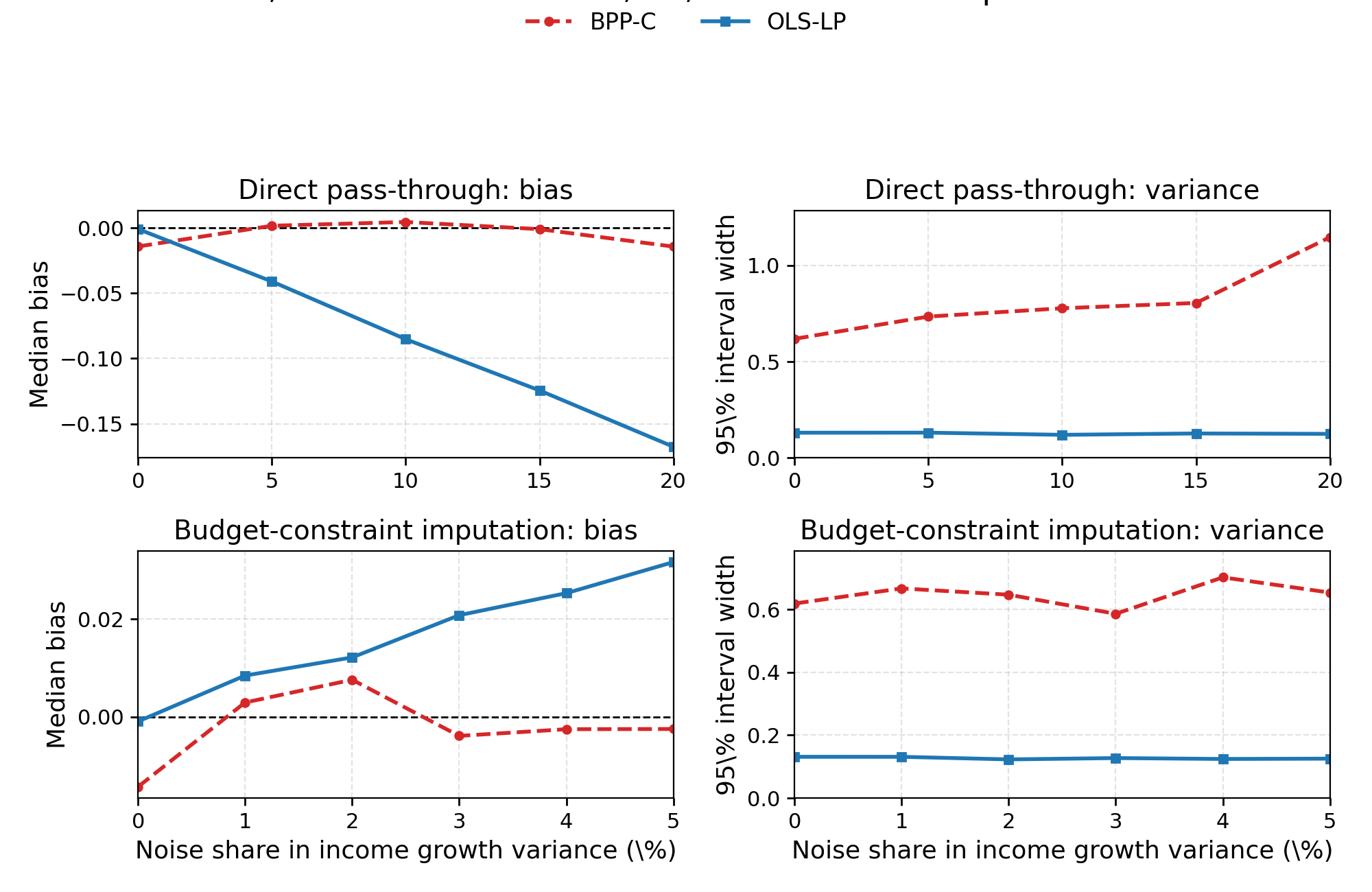}
    \figurenote{Each point uses 1,000 simulated samples with 30,421 observations per sample.}
\end{figure}

\subsection{Independent Measurement Error}
The first simulation corresponds to a survey setting in which income and consumption are measured separately. Observed income growth is contaminated by an i.i.d. income error $\xi_t-\xi_{t-1}$, while observed consumption growth has an independent measurement error:
\begin{align}
    \Delta \log y^{obs}_t
    =
    \eta_t+\varepsilon_t-(1-\theta_1)\varepsilon_{t-1}
    -\theta_1\varepsilon_{t-2}
    +\xi_t-\xi_{t-1},
    \label{eq:measurement_error_income_growth}\\
    \Delta \log c^{obs}_t
    =
    \lambda\eta_t+\gamma\varepsilon_t+\zeta_t-\zeta_{t-1}.
    \label{eq:independent_measurement_error_consumption_growth}
\end{align}
The BPP-C estimator remains centered on $\gamma$ because the two-period-ahead instrument is not contaminated by the contemporaneous MA(1) measurement-error term. The cost is precision: as the noise share rises, the transitory-shock variance consistent with the observed income moments falls, weakening the long-horizon first stage. The naive projection estimator is instead attenuated toward zero because income measurement error inflates the variance of the recovered transitory signal without creating a corresponding covariance with consumption growth.

\subsection{Budget-Constraint Imputation}
We impute consumption from income and wealth using the household budget constraint. This means that any measurement error in income enters observed consumption mechanically:
\begin{equation}
    \Delta \log c^{obs}_t
    =
    \lambda\eta_t+\gamma\varepsilon_t+\zeta_t-\zeta_{t-1}
    +\xi_t-\xi_{t-1}.
\label{eq:budget_constraint_measurement_error_consumption_growth}
\end{equation}
The same $\xi_t$ therefore appears in both observed income growth \eqref{eq:measurement_error_income_growth} and observed consumption growth \eqref{eq:budget_constraint_measurement_error_consumption_growth}. BPP-C remains robust because the contamination is concentrated at short lags. The naive projection estimator now faces two forces. The contaminated income history still attenuates the recovered shock, but the same contamination also enters imputed consumption and raises the numerator of the pass-through regression. In the simulation, the second force dominates over the register-relevant range, so the naive projection estimator overstates $\gamma$ as the income noise share rises.

\section{Data Appendix}
\label{app:data}

\subsection{Register Sources and Linkages}

The longitudinal dataset LISA (Longitudinal Integrations Database for Sickness Insurance and Labour Market Studies) is provided by Statistics Sweden and merges several administrative and tax registers for the universe of Swedish individuals aged 16 and above. LISA contains annual information on labor earnings, capital income, transfers, taxes, and socio-demographic variables such as age, education, occupation, marital status, and municipality of residence.

We link LISA to complementary asset registers. The wealth register (\emph{F\"orm\"ogenhetsregistret}) covers the period 1999--2007, when Sweden levied a comprehensive wealth tax and the tax administration collected balance-sheet information for the entire population. It reports end-of-year balance values for financial assets, listed securities, real estate, and liabilities. On the asset side, it provides aggregate positions by instrument, such as bank deposits, bonds, equities, and mutual funds. On the liability side, it includes total outstanding debt, including mortgages, consumer credit, and student loans. Financial institutions were legally required to report this information directly to the tax authority for the administration of the wealth tax.

The KURU register provides more disaggregated security-level information. It records holdings by International Securities Identification Number (ISIN), including the quantity held for securities that Statistics Sweden can match to a market price. Housing registers include \textit{Fastighetsprisstatistiken}, which records transaction prices for owned real estate, and KU55, which records co-operative-apartment transactions. These sources are merged to the individual and household level using personal identification numbers and dwelling identifiers.

\subsection{Consumption Imputation}

We impute consumption from the household budget constraint,
\begin{equation}
    C_{i,t}=Y_{i,t}-S_{i,t},
    \label{eq:BC}
\end{equation}
where $C_{i,t}$ is consumption, $Y_{i,t}$ is disposable income, and $S_{i,t}$ is savings. To construct our measure of disposable income, we start from the LISA measure of income, which includes employment-related earnings, including wages, vacation pay, severance payments, and wage income for self-employed workers and business owners as well as capital income. Business losses for small business owners and independent contractors are included from 2004 onward. We add government transfers, including unemployment benefits, sickness and disability benefits, pension benefits, housing benefits, child benefits, and smaller transfer programs, and deduct all taxes. We subtract the change in cumulative student debt, since these are automatic payments, and capital income, since it is endogenous to the consumption--savings decision and would make identification of transitory shocks problematic. Our final measure of disposable income is the same as in \citet{Kolsrud_et_al_2020_JPUB}, except that our subtraction of capital income is based on the summary measure ``kapinc'', net of taxes, whereas they use the gross value of the more narrow measure ``KIRANTA'', which does not include capital gains.

Conceptually, savings should capture \emph{active} rebalancing---net purchases of assets and net debt repayment---and should not mechanically reflect \emph{passive} changes in wealth due to asset price movements. This distinction matters because the registers typically report end-of-year positions, sometimes in quantities and sometimes only in values, while transactions and intra-year trading are not directly observed. We therefore use two complementary approaches.

\paragraph{Quantity approach.}
When end-of-year holdings are observed in quantities $A_{i,k,t}$ for asset $k$ and a corresponding market price $p_{k,t}$ is available, net purchases are
\[
S^{Q}_{i,k,t}=p_{k,t}\left(A_{i,k,t}-A_{i,k,t-1}\right).
\]
An increase in holdings during year $t$ is treated as saving and therefore reduces imputed consumption, whereas a reduction in holdings is treated as dissaving and increases imputed consumption. If the household holds the same quantity over the year, $S^{Q}_{i,k,t}=0$ even when prices move. The main limitation is that quantities are observed only at the end of the year, so this approach values the net annual change in holdings at the end-of-year price and abstracts from the timing of trades within the year.

\paragraph{Price approach.}
When only balance values $W_{i,k,t}=p_{k,t}A_{i,k,t}$ are observed, or asset-specific prices are unavailable, we infer net purchases by subtracting passive valuation changes:
\[
	S^{P}_{i,k,t}=\Delta W_{i,k,t} -\frac{\Delta p_{k,t}}{p_{k,t-1}} W_{i,k,t-1}.
\]
Here $\Delta X_t \equiv X_t-X_{t-1}$. The term $(\Delta p_{k,t}/p_{k,t-1})W_{i,k,t-1}$ is the change in wealth that would arise purely from applying the asset price movement to last year's position. The residual $S^{P}_{i,k,t}$ is interpreted as the household's net purchase flow. If an asset-specific price is unavailable, we use a category-level price index as in \cite{Kolsrud_et_al_2020_JPUB}. When quantities and prices are observed consistently for the same security, the two approaches are algebraically equivalent; in practice they differ because some holdings are observed only in aggregates or prices are not matchable at the relevant level of detail.

\paragraph{Financial assets.}
Fixed-price instruments, such as bank accounts, money-market funds, and similar products, enter through changes in reported end-of-year balances. For bank accounts, consumption out of these assets is inferred from balance changes between years $t-1$ and $t$, together with interest income on deposits, which is already included in the income measure. Liabilities are treated analogously: changes in outstanding debt, net of interest payments that are already part of the income measure, are interpreted as consumption financed through borrowing or saving through debt repayment.

For stocks, bonds, and mutual funds, we supplement the wealth register with KURU. Swedish financial institutions must report each client's ISIN-level holdings and the quantity held. Statistics Sweden registers securities for which it can match a market price, so the register omits non-listed assets such as private business equity or shares in closely held, non-traded companies. Existing evidence, see \cite{Bach2020}, suggests that such unlisted securities are negligible for the vast majority of households in Sweden and only represent a substantial share of wealth for a very small fraction of the population at the absolute top of the wealth distribution. For listed financial instruments, we use closing prices from LSEG Refinitiv Eikon Datastream on the last trading day of December. Datastream covers a broad set of instruments across major asset classes, including active and delisted exchange-traded instruments. For securities in KURU that cannot be matched to a price, we use the price approach with an aggregate price index.

For options, warrants, and tax-favoured capital-insurance accounts, we observe the total balance rather than quantities. We therefore apply a price-based method with an aggregate price index to impute the associated consumption flows. For tax-deferred private pension accounts, we observe contributions and withdrawals but not the underlying balance. Since the object of interest is the consumption flow, these inflows and outflows are sufficient to construct the relevant series.

For the stock of total financial wealth, we use the maximum of financial wealth extracted from Datastream and financial wealth in the wealth register. This uses the most accurate source for each individual, as it minimizes the problem of omitting the value of instruments whose value was missing in the wealth registry. Datastream financial wealth is substantially larger than wealth-register financial wealth, mostly because Datastream covers foreign-traded stocks and other instruments that Statistics Sweden did not fully value in the wealth register. The difference is concentrated among wealthy individuals; after trimming the top one percent by financial wealth, Datastream financial wealth remains roughly twice as large as the wealth-register measure. The stock of total financial wealth is used to stratify the population in heterogeneity analysis, not to impute consumption.

\paragraph{Real assets.}
We separate real assets into owned real estate and co-operative apartments. Owned real estate includes stand-alone houses, semi-detached houses, bungalows, terrace houses, second homes, commercial real estate, apartment buildings, industrial plants, and farming properties. Co-operative apartments are treated separately because Swedish apartment residents commonly buy a share in a housing co-operative rather than legally owning the property itself.

For owned real estate, we use Fastighetsprisstatistiken to observe acquisitions and sales and classify transactions into primary housing, second homes, farms, and commercial real estate. Let $K^H_{i,t}$ denote acquisitions and $S^H_{i,t}$ sales. The real-estate saving flow entering the budget constraint is
\[
    \Delta H_{i,t}=K^H_{i,t}-S^H_{i,t},
\]
so net purchases reduce imputed consumption and net sales increase it. The implementation differs by period because registry detail changes. For 2004 onward, all parties involved in each transfer are observed, so transaction values can be assigned at the individual level using party-level participation. For 2000--2003, party information is more limited, so household-level transaction values are allocated using ownership shares from the wealth register, using contemporaneous shares on the purchase side and lagged shares on the sale side.

For co-operative apartments, KU55 records sales with transferred shares and sale values, which we use to construct individual-level sale proceeds. These are combined with housing flows to construct a co-op-adjusted housing flow from 2004 onward. For 2005--2007, we additionally impute missing co-op purchase values for individuals with persistent positive co-op wealth in adjacent years. The imputation has two steps. First, we estimate a purchase-probability model using household demographics, disposable income, changes in household size, debt and co-op wealth changes, year effects, age-bin controls, regional controls, and indicators for related housing transactions. Second, conditional on observables, we estimate expected purchase prices and form an imputed co-op flow as the difference between observed sale values and predicted purchase values for individuals with sufficiently high predicted purchase probability. The imputed flow is scaled by observed ownership shares and merged back to the person-year panel.

We also account for housing services through imputed rents. Following the approach used by Statistics Sweden for national accounts, imputed rent is based on potential rental income from subletting the house, computed as square-metre living space multiplied by a square-metre rental value that varies by region, construction period, and dwelling category. Since individual imputed rents are not observed in our data, we assign aggregate imputed rents to households based on the value of their primary real estate in the wealth tax register.

\subsection{Consumption Validation}
\label{app:data_consumption_validation}

Figure~\ref{fig:aggregate_cons} compares our registry-based per-capita consumption measure to the national-accounts series and the Household Budget Survey (HUT). The national-accounts series is obtained by taking nominal household consumption expenditure under ESA2010, in current SEK millions by durability and year, and dividing by the number of individuals in MONA. The HUT series is obtained by dividing household consumption by household size and weighting observations by survey weights. All series are expressed in 2003 SEK.

The registry measure lies between the survey and national-accounts series throughout the sample period. The survey provides direct measures of bi-weekly consumption expenditures at the time the household is surveyed. In the national accounts, consumption is computed from several sources, including HUT survey information, registry information, and VAT receipts. A main difference relative to the national accounts is that some investments are treated as expenditures in our measure and as investment in the national accounts. In contrast with survey data, both the national accounts and the registry measure include imputed rents. The three series are therefore broadly comparable in levels over the sample period.

\begin{figure}[t]
    \centering
    \caption{Registry consumption, national accounts, and household budget survey (HUT)}
    \label{fig:aggregate_cons}
    \includegraphics[width=0.86\textwidth]{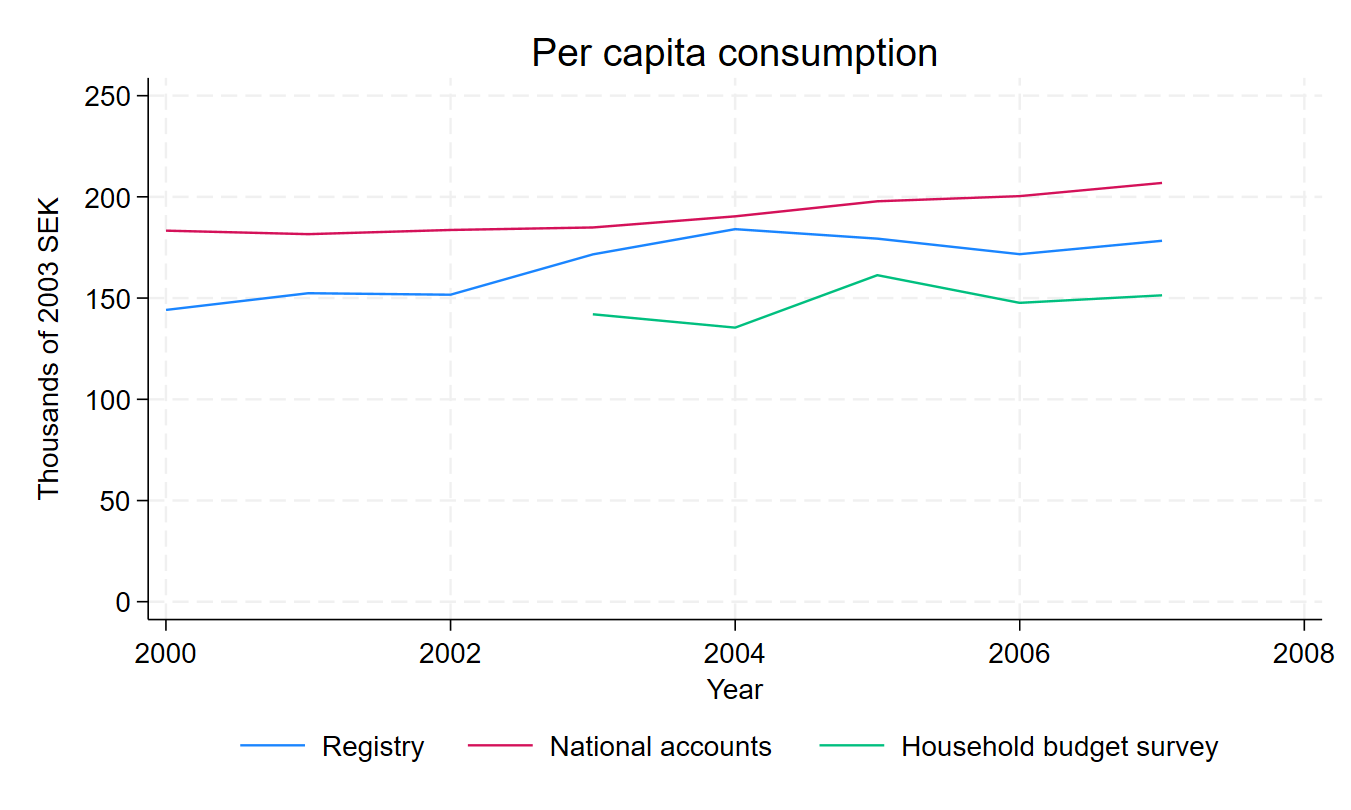}
    \figurenote{The figure reports per-capita consumption. All series are in 2003 SEK.}
\end{figure}

\subsection{Income and Wealth Variables}
\label{sec:vars_def}

\paragraph{Disposable income ($Y_{i,t}$).}
Annual disposable income sums labor earnings and government transfers, net of taxes. Labor earnings include wages, vacation pay, severance payments, and wage income for self-employed workers and business owners. For the income measure used in the income shocks estimation, we subtract capital income because it is endogenous to consumption and saving decisions.

\paragraph{Consumption ($C_{i,t}$).}
Registry consumption is imputed from Equation~\eqref{eq:BC} using the savings flows described above.

\paragraph{Liquid wealth ($A^{\mathrm{liq}}_{i,t}$).}
Liquid wealth is financial wealth plus checking-account balances plus capital-insurance wealth.

\paragraph{Net illiquid wealth ($A^{\mathrm{illiq}}_{i,t}$).}
Net illiquid wealth is real estate assets minus total debt and cumulative student debt.

\paragraph{Net total wealth ($A^{\mathrm{net}}_{i,t}$).}
Net total wealth is financial wealth plus checking-account balances plus capital-insurance wealth plus real estate assets, minus total debt and cumulative student debt:
\[
    A^{\mathrm{net}}_{i,t}
    =
    A^{\mathrm{fin}}_{i,t}
    + A^{\mathrm{checking}}_{i,t}
    + A^{\mathrm{capins}}_{i,t}
    + A^{\mathrm{real}}_{i,t}
    - \mathrm{Debt}_{i,t}
    - \mathrm{StudentDebt}_{i,t}.
\]

\paragraph{Cash-on-hand ($M_{i,t}$).}
Cash-on-hand at the start of year $t$ is
\[
M_{i,t}=(1+r)A^{\mathrm{liq}}_{i,t}+Y_{i,t}.
\]
The term $(1+r)A^{\mathrm{liq}}_{i,t}$ is the end-of-year-$t-1$ liquid-asset position observed in the registers; returns over $t-1$ are already embedded in this stock, so the interest rate $r$ is not observed or imputed separately.

\paragraph{Normalized cash-on-hand ($m_{i,t}$).}
Normalized cash-on-hand divides cash-on-hand by smoothed permanent income:
\[
m_{i,t}=\frac{M_{i,t}}{P_{i,t}},
\]
where $P_{i,t}$ is the permanent-income component recovered from the estimated income process. 

\subsection{Sample Construction and Restrictions}

The imputation of consumption is carried out on the full population. For the main analysis, we restrict the sample to individuals aged 30--65 with non-missing income, consumption, and wealth data. We drop observations with non-positive imputed consumption. We also exclude observations with extreme income or consumption growth rates.

\paragraph{Singles (baseline).}
The baseline \emph{Singles} sample comprises one-person households, excluding married or registered-partner cases, that satisfy the age and data-quality restrictions above.

\paragraph{Households.}
The \emph{Households} sample includes household heads, defined as the highest-income earner, in stable household structures. Income, consumption, and wealth are aggregated to the household level.

\paragraph{Housing and co-operative-apartment exclusions.}
For both subsamples, we exclude observations involved in housing or co-operative-apartment transactions before estimating Equation \eqref{eq:second_stage_estimation}. This restriction aligns the sample with the identifying restriction in the main text by focusing on periods without active illiquid-asset adjustment.

\paragraph{Detrended growth rates.}
For estimation, we prepare balanced lead-lag moments and detrended growth components. We residualize log income and log consumption using flexible controls: time effects, cohort, family composition, employment status, education-time interactions, region effects, and lagged controls. The corresponding growth rates are then built from first differences and enter the income-process estimation and the second-stage MPC regressions.

\section{Additional Results and Robustness}
\label{app:additional_results}

\subsection{Household and MA(2) Income-Process Estimates}

Table~\ref{tab:income_process_ma2} reports the MA(2) income-process estimates used in the robustness exercises. The estimated second MA parameter is small relative to the first: $\theta_2=0.0694$ for Singles and $\theta_2=0.0642$ for Households. The income-growth autocovariance at the third lead is correspondingly small. The MA(2) specification therefore allows a richer transitory-income footprint, but it does not point to a fundamentally different income process.

\begin{table}[t]
    \centering
    \caption{MA(2) income-process moments and parameter estimates}
    \label{tab:income_process_ma2}
    \begin{minipage}[t]{0.48\textwidth}
        \centering
        \caption*{A. Singles: moments}
        \resizebox{\linewidth}{!}{\begin{tabular}{l c c c c}
\hline\hline
 & $\Delta \ln y_t$ & $\Delta \ln y_{t+1}$ & $\Delta \ln y_{t+2}$ & $\Delta \ln y_{t+3}$ \\
\hline
$\Delta \ln y_t$ & \shortstack[c]{0.0297\\(0.0295, 0.0299)} & \shortstack[c]{-0.0072\\(-0.0073, -0.0070)} & \shortstack[c]{-0.0026\\(-0.0027, -0.0025)} & \shortstack[c]{-0.0009\\(-0.0010, -0.0008)} \\
$\Delta c_t$ & \shortstack[c]{0.0214\\(0.0212, 0.0216)} & \shortstack[c]{-0.0035\\(-0.0037, -0.0034)} & \shortstack[c]{-0.0023\\(-0.0024, -0.0021)} & \shortstack[c]{-0.0006\\(-0.0008, -0.0005)} \\
\hline
\(N\) & \multicolumn{4}{c}{   1,155,414} \\
\hline\hline
\end{tabular}
}
    \end{minipage}\hfill
    \begin{minipage}[t]{0.48\textwidth}
        \centering
        \caption*{B. Singles: parameters}
        \resizebox{\linewidth}{!}{\begin{tabular}{l c c c c}
\hline\hline
 & $\theta_1$ & $\theta_2$ & $\sigma^2_{\epsilon}$ & $\sigma^2_{\eta}$ \\
\hline
Estimate & \shortstack[c]{0.3056\\(0.2982, 0.3130)} & \shortstack[c]{0.0694\\(0.0636, 0.0751)} & \shortstack[c]{0.0142\\(0.0140, 0.0145)} & \shortstack[c]{0.0077\\(0.0074, 0.0080)} \\
\(N\) & \multicolumn{4}{c}{   1,155,414} \\
\hline\hline
\end{tabular}
}
    \end{minipage}

    \vspace{1em}

    \begin{minipage}[t]{0.48\textwidth}
        \centering
        \caption*{C. Households: moments}
        \resizebox{\linewidth}{!}{\begin{tabular}{l c c c c}
\hline\hline
 & $\Delta \ln y_t$ & $\Delta \ln y_{t+1}$ & $\Delta \ln y_{t+2}$ & $\Delta \ln y_{t+3}$ \\
\hline
$\Delta \ln y_t$ & \shortstack[c]{0.0288\\(0.0286, 0.0289)} & \shortstack[c]{-0.0069\\(-0.0070, -0.0069)} & \shortstack[c]{-0.0026\\(-0.0027, -0.0026)} & \shortstack[c]{-0.0008\\(-0.0009, -0.0008)} \\
$\Delta c_t$ & \shortstack[c]{0.0218\\(0.0216, 0.0219)} & \shortstack[c]{-0.0036\\(-0.0037, -0.0035)} & \shortstack[c]{-0.0024\\(-0.0025, -0.0023)} & \shortstack[c]{-0.0006\\(-0.0007, -0.0005)} \\
\hline
\(N\) & \multicolumn{4}{c}{   2,150,175} \\
\hline\hline
\end{tabular}
}
    \end{minipage}\hfill
    \begin{minipage}[t]{0.48\textwidth}
        \centering
        \caption*{D. Households: parameters}
        \resizebox{\linewidth}{!}{\begin{tabular}{l c c c c}
\hline\hline
 & $\theta_1$ & $\theta_2$ & $\sigma^2_{\epsilon}$ & $\sigma^2_{\eta}$ \\
\hline
Estimate & \shortstack[c]{0.3041\\(0.2988, 0.3094)} & \shortstack[c]{0.0642\\(0.0601, 0.0683)} & \shortstack[c]{0.0137\\(0.0136, 0.0139)} & \shortstack[c]{0.0076\\(0.0074, 0.0078)} \\
\(N\) & \multicolumn{4}{c}{   2,150,175} \\
\hline\hline
\end{tabular}
}
    \end{minipage}
\end{table}

\subsection{MA(2) MPC Profiles}

Under MA(2), the state is $(m_{i,t-1},\varepsilon_{i,t-1},\varepsilon_{i,t-2})$. We partition observations into a $10\times5\times2$ grid and average across the two lagged-shock dimensions to obtain cash-on-hand decile profiles. Figure~\ref{fig:mpc_gradient_ma2} shows that the MA(2) estimates preserve the main result. The projection estimates decline sharply with cash-on-hand in both samples, and the future-income IV estimates remain substantially less precise. For the BPP-C estimates in this figure, we report (incorrect) non-bootstrap confidence intervals because bootstrapping the full MA(2) BPP-C estimation procedure is computationally very demanding.

\begin{figure}[t]
    \centering
    \caption{MPC by cash-on-hand, MA(2)}
    \label{fig:mpc_gradient_ma2}
    \includegraphics[width=\textwidth]{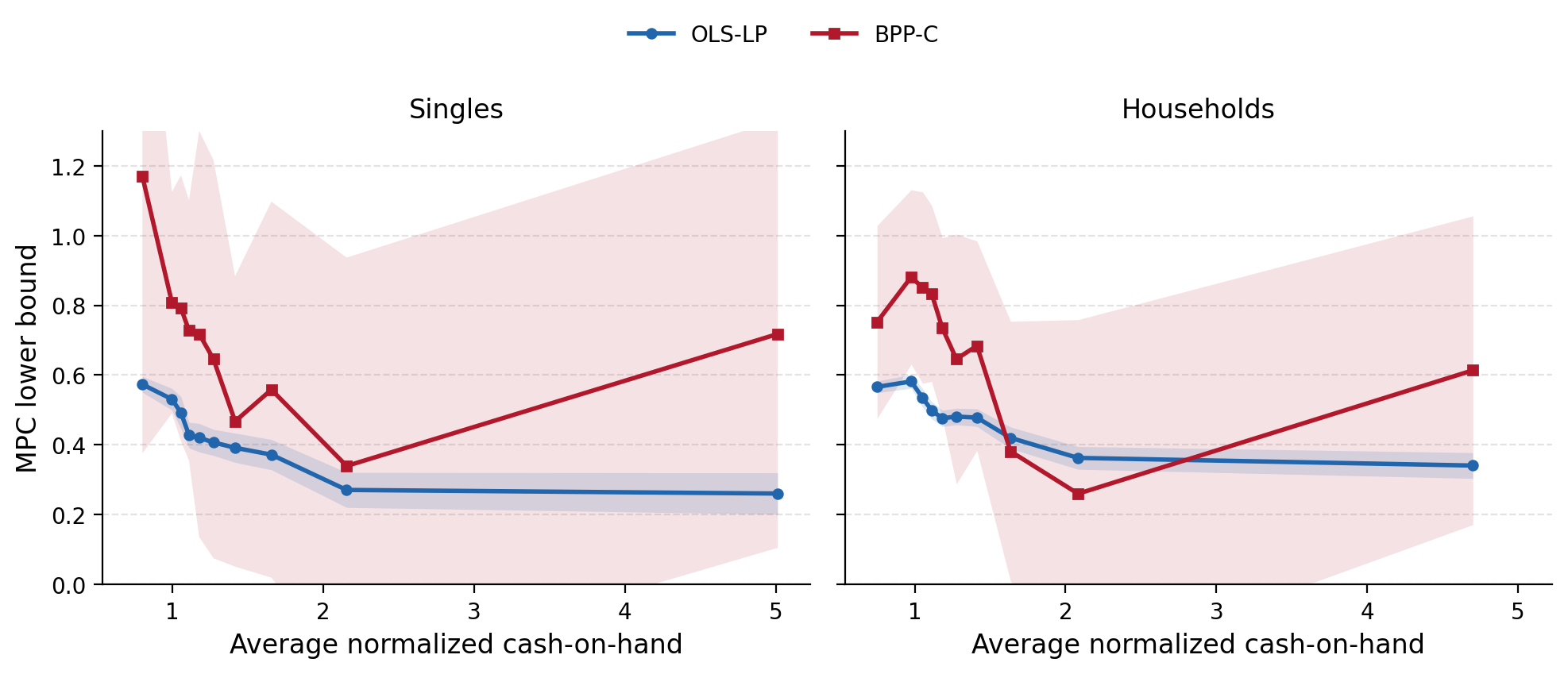}
    \figurenote{Estimates show lower bounds of the MPC, and are computed within decile bins of lagged normalized cash-on-hand; the horizontal axis plots the observation-count-weighted average normalized cash-on-hand in each bin. The projection estimator uses the Kalman-smoothed transitory shock and reports 95 percent confidence intervals based on bootstrap standard errors clustered at the household level. BPP-C uses the distant future-income-growth instrument and reports the (incorrect) conventional non-bootstrap confidence intervals because bootstrapping the full MA(2) BPP-C estimation procedure is computationally very demanding.}
\end{figure}

\subsection{Pass-Through Profiles}

Figure~\ref{fig:pass_through_profiles_ma2} reports the estimated MA(2) pass-through coefficients before the MPC conversion, corresponding to the baseline MA(1) profiles in Figure~\ref{fig:pass_through_profiles_ma1}. The transitory pass-through remains downward sloping in cash-on-hand for both Singles and Households, while the permanent-shock pass-through is comparatively flat. Thus the richer income-process specification preserves the pass-through pattern that generates the MPC gradient.

\begin{figure}[t]
    \centering
    \caption{Pass-through by cash-on-hand, MA(2)}
    \label{fig:pass_through_profiles_ma2}
    \includegraphics[width=0.92\textwidth]{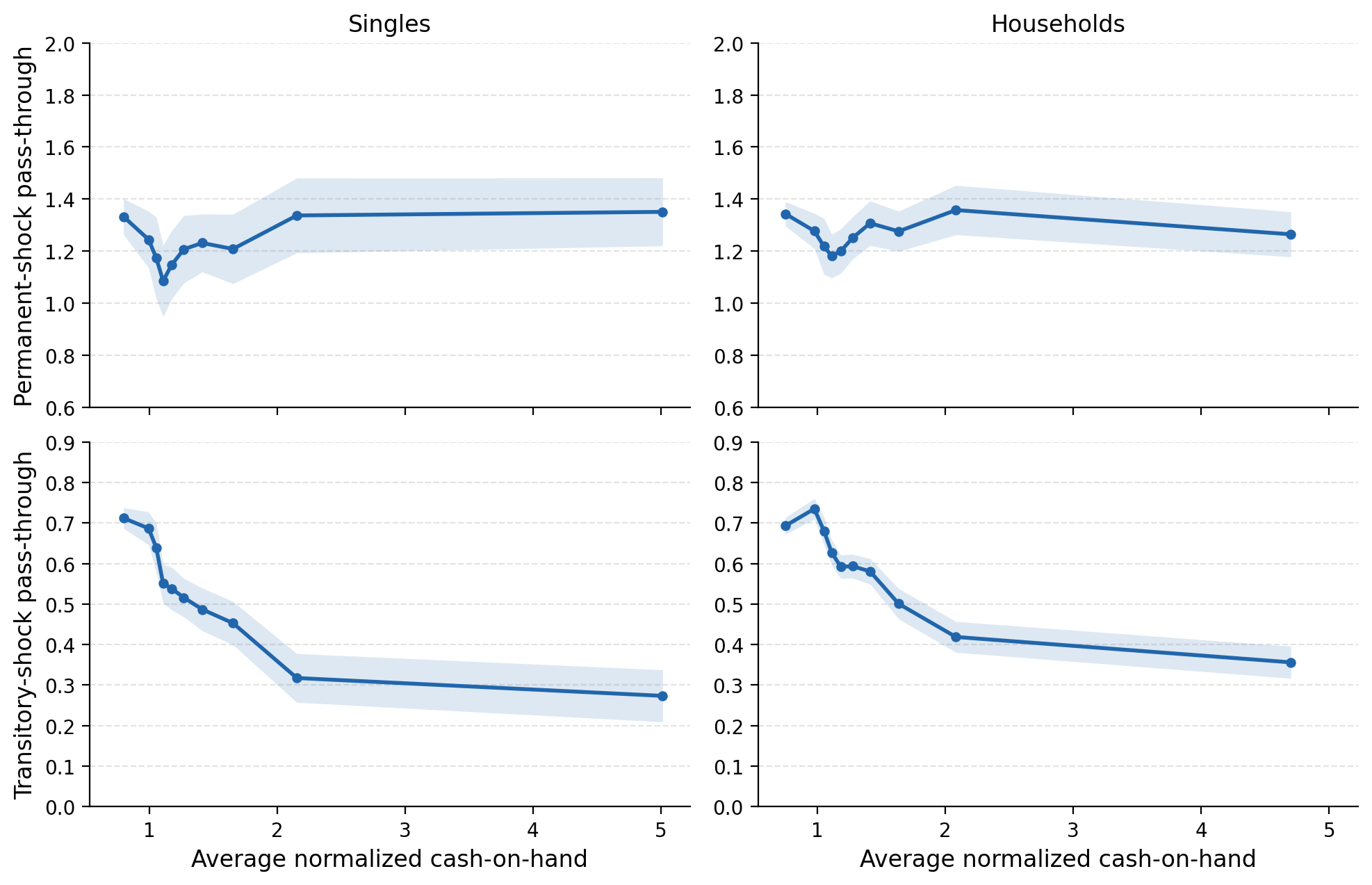}
    \figurenote{Columns report Singles and Households; rows report permanent- and transitory-shock pass-through coefficients from the OLS projection estimator. Estimates are computed within decile bins of lagged normalized cash-on-hand; the horizontal axis plots the observation-count-weighted average normalized cash-on-hand in each bin. Shaded regions show 95 percent confidence intervals based on bootstrap standard errors clustered at the household level.}
\end{figure}

\subsection{Comparison to BPP-C with bias-correction}
\label{app:bpp_c_near_lead_bias_correction}

This final comparison asks how much precision is lost by imposing the robust BPP-C lead restriction relative to a near-lead restriction. For a future-income horizon $h$, define the future-income instrument and the IV estimator as
\[
    Z_{i,t}(h)=\Delta\log Y_{i,t+h},
    \qquad
    \gamma^{IV}(h)
    =
    \frac{\operatorname{Cov}(Z_{i,t}(h),\Delta\log C_{i,t})}
         {\operatorname{Cov}(Z_{i,t}(h),\Delta\log Y_{i,t})}.
\]
With an MA$(k)$ transitory income component, the BPP-C restriction uses the distant lead $h=k+1$: $h=2$ under MA(1) and $h=3$ under MA(2). This lead filters out the permanent shock and the lagged transitory shocks that may enter consumption growth, which makes it robust to misspecification of the consumption equation. The cost is that the identifying covariance is proportional to the last MA coefficient, so the instrument becomes weak when transitory persistence is modest.

A shorter lead is more precise but is mechanically biased under the MA income process. However, this bias can be corrected knowing the parameters of the income process. Under the baseline MA(1) process and the state-conditioned pass-through equation used in the main text,
\[
    \operatorname{plim}\gamma^{IV}(1)
    =
    \frac{\gamma_0}{1-\theta_1}.
\]
and thus the bias-corrected estimate is
\[
    \widehat{\gamma}^{IV,1,\mathrm{bc}}_0
    =
    (1-\widehat{\theta}_1)\widehat{\gamma}^{IV}(1).
\]
For the MA(2) robustness specification, Figure~\ref{fig:bpp_near_distant} reports the one-lead bias-corrected estimate. The correction multiplies the one-lead estimate by $B_1(\widehat{\theta}_1,\widehat{\theta}_2)/(1-\widehat{\theta}_1)$, where $B_1(\theta_1,\theta_2)\equiv-\operatorname{Cov}(\Delta\log Y_{i,t+1},\Delta\log Y_{i,t})/\operatorname{Var}(\varepsilon_{i,t})$.

This bias correction is not a generic replacement for BPP-C. It relies on a correctly specified state-conditioned consumption equation: lagged transitory shocks must be included in the predetermined state rather than enter the omitted part of consumption growth. It also relies on negligible income measurement error. With classical measurement error in income, the near-lead covariance contains an additional noise term, and the same scalar correction no longer recovers $\gamma_0$ unless that noise variance is known. Figure~\ref{fig:bpp_near_distant} should therefore be read as a comparison between the precise, bias-corrected near-lead instrument under these maintained assumptions and the more robust but weaker distant-lead BPP-C instrument. We use (incorrect) standard asymptotic confidence intervals in this comparison because bootstrapping the full set of near- and distant-lead BPP-C estimates is computationally very demanding.

\begin{figure}[t]
    \centering
    \caption{Near- versus distant-lead future-income IV estimates}
    \label{fig:bpp_near_distant}
    \includegraphics[width=0.86\textwidth]{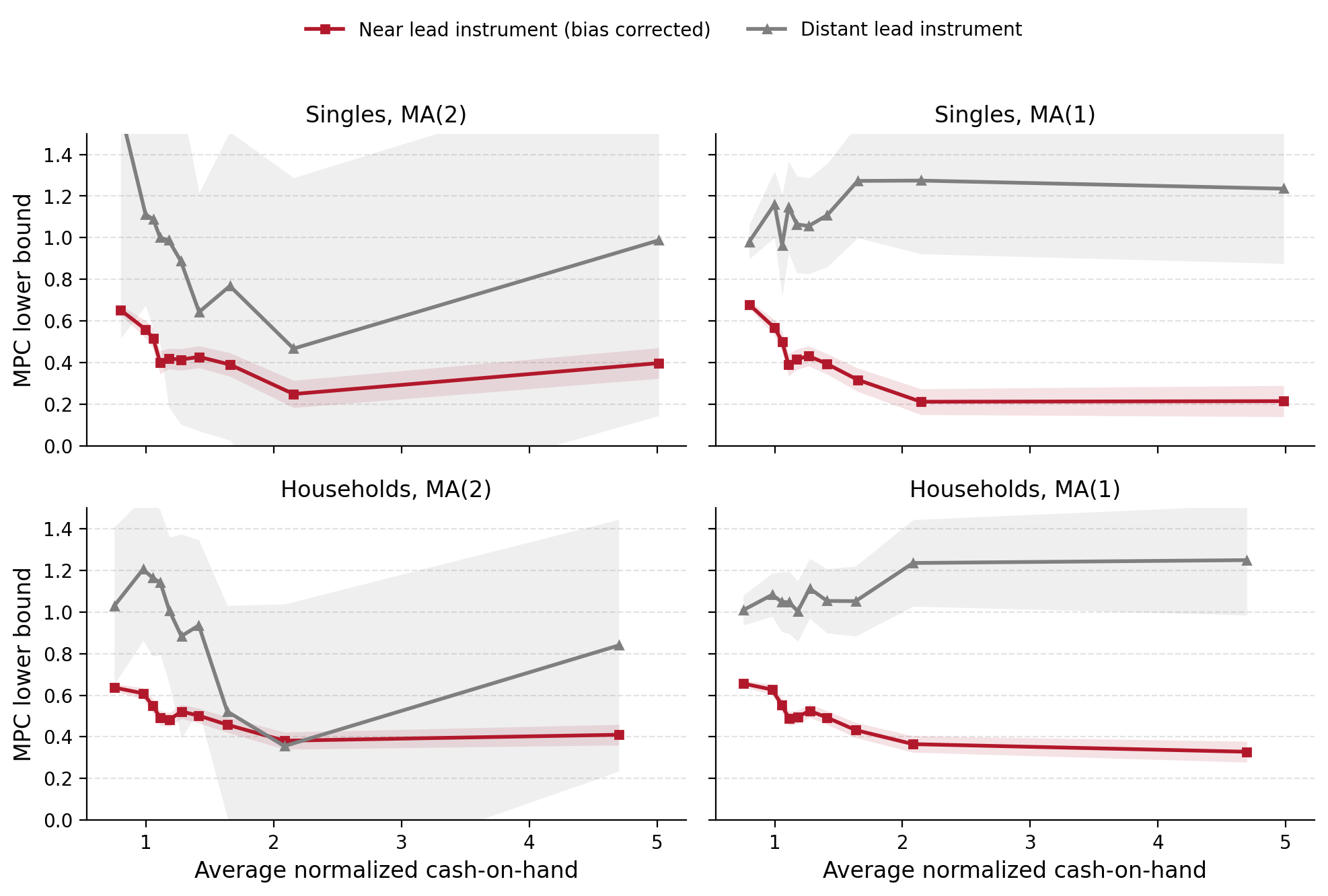}
    \figurenote{Estimates show lower bounds of the MPC, and are computed within decile bins of lagged normalized cash-on-hand; the horizontal axis plots the observation-count-weighted average normalized cash-on-hand in each bin. The near-lead instrument applies the analytical bias correction under the state-conditioned pass-through equation, using the +1 lead in both MA(1) and MA(2). The distant-lead instrument follows the robust BPP-C restriction, using the +2 lead under MA(1) and the +3 lead under MA(2). Shaded regions use (incorrect) standard asymptotic confidence intervals for all gradients because bootstrapping the full set of near- and distant-lead BPP-C estimates is computationally very demanding.}
\end{figure}

\renewcommand{\refname}{Appendix References}
\putbib[References/ref]
\end{bibunit}
\end{document}